\documentclass[final,14pt]{article}
\usepackage{verbatim}
\usepackage{latexsym}
\usepackage{amsmath,amsthm,amsfonts}
\usepackage{eucal}
\usepackage{cases}

\usepackage{mathtext}
\usepackage[cp1251]{inputenc}
\usepackage[T2A]{fontenc}
\usepackage[dvips]{graphicx}
\usepackage{amsmath}
\usepackage{amssymb}
\usepackage{amsxtra}
\usepackage[sort&compress,numbers]{natbib}
\usepackage{latexsym}
\usepackage{ifthen}
\usepackage{color}
\usepackage{mathtools}
\usepackage[colorlinks=true,linkcolor=blue]{hyperref}
\usepackage{showkeys}

\newtheorem{theorem}{Theorem}
\newtheorem{defi}{Definition}[section]

\newtheorem{proposition}{Proposition}[section]
\newtheorem{remark}{Remark}[section]

\newtheorem{corollary}{Corollary}[section]
\newtheorem{lem}{Lemma}[section]
\numberwithin{equation}{section}

\newcommand{\eps}{\epsilon}
\newcommand{\lb}{\label}
\newcommand{\go}{\rightarrow}

\newcommand{\ee}{\end{equation}}
\newcommand{\be}{\begin{equation}}
\newcommand{\bea}{\begin{eqnarray}}
\newcommand{\eea}{\end{eqnarray}}
\newcommand{\sbea}{\begin{subequations}\begin{eqnarray}}
\newcommand{\seea}{\end{eqnarray}\end{subequations}} 
\newcommand{\ees}{\end{equation*}}
\newcommand{\bes}{\begin{equation*}}
\newcommand{\beas}{\begin{eqnarray*}}
\newcommand{\eeas}{\end{eqnarray*}}

\newcommand{\rf}[1]{(\ref{#1})}

\newcommand{\mR}{\mathbb{R}\,}
\newcommand{\N}{\mathbb{N}\,}

\newcommand{\la}{\lambda_n\,}

\newcommand {\Rl} {\mathbb{R}} \newcommand {\Z} {\mathbb{Z}}
 \newcommand {\mA} {\mathcal{A}}

\newcommand {\D} { \Delta}

\newcommand {\dist} { \mbox{dist}}

\DeclarePairedDelimiter{\floor}{\lfloor}{\rfloor}

\begin{document}

\title{Surface Energies Arising in Microscopic Modeling of Martensitic Transformations}
\author{Georgy Kitavtsev\footnote{Max Planck Institute for Mathematics in the Sciences, Inselstr. 22, D--04103 Leipzig, Germany. {\tt E-mail}: Georgy.Kitavtsev@mis.mpg.de}, Stephan Luckhaus\footnote{Mathematical Institute, University of Leipzig, 04009 Leipzig, Germany. {\tt E-mail}: stephan.luckhaus@math.uni-leipzig.de},
and Angkana R\"uland\footnote{Mathematical Institute, University of Bonn, Endenicher Allee 60, D - 53115 Bonn, Germany. {\tt E-mail}: rueland@math.uni-bonn.de}}
\date{\today}
\maketitle
\begin{abstract}
In this paper we construct and analyze a two-well Hamiltonian on a 2D atomic lattice. The two wells of the Hamiltonian
are prescribed by two rank-one connected martensitic twins, respectively. By constraining the deformed configurations
to special 1D atomic chains with position-dependent elongation vectors for the vertical direction, we show that the structure of ground states
under appropriate boundary conditions is close to the macroscopically expected
twinned configurations with additional boundary layers localized
near the twinning interfaces. In addition, we proceed to a continuum limit, show asymptotic piecewise rigidity of minimizing sequences
and rigorously derive the corresponding limiting form of the surface energy.
\end{abstract}

\section{Introduction}
In the last decades there has been an intensive mathematical research on martensitic transformations
in shape memory alloys using nonlinear elasticity models of continuum mechanics, see e.g.~\cite{BJ87,KM92,DM95,B03}.
In several models a finite length scale of the emerging martensitic microstructure was obtained and analyzed 
by adding penalizing higher order gradient terms to the elastic energy~\cite{KM92,DM95,CS06}. 
In parallel to this, the analysis of microscopic models in nonlinear elasticity and the systematic derivation of the corresponding discrete-to-continuum limits has recently
attracted a lot of attention ~\cite{BBL02,TF02,AC04,BG06,EM07,BS12}.  
In this context, even more general reference-free models have been constructed and analyzed in~\cite{Th06,LM10}.
In particular, the derivation of the arising limiting surface energies was
investigated rigorously in~\cite{CT02,BC07,Th11,SSZ11,Hu13,Ro12}.
While the cases when the minimizers of the energy belong to a single well
are well understood also in several dimensions, see e.g.~\cite{Th11}, 
surface energies for two-well, discrete problems have to the best of our knowledge only been derived rigorously in 1D cases~\cite{BC07,SSZ11}.
The multi-well structure, however, is an intrinsic feature of martensitic microstructures 
and of the understanding of the appearing characteristic, finite length scales. \\

In this paper we investigate the problem of the formation of twinned martensitic microstructures from a microscopic point of view. 
We begin by defining a class of atomistic two-well Hamiltonians
on a 2D atomic lattice. These Hamiltonians feature nonconvex interactions and are constructed to model simple martensitic microstructures.
We aim at describing the structure of their ground states and at deriving a limiting form of the corresponding surface energies at zero temperature. 
The latter should emerge naturally from the full microscopic energy given by the Hamiltonian.\\

The novelty of our approach consists of the \emph{non-discreteness} of the set constituting our minimizers. 
The energy wells are given by $SO(2)U_{0}\cup SO(2)U_{1}$ with $U_{0}$ and $U_1$ being rank-one-connected matrices in $SL(2,\mR)$.
This setting allows for a very rich microscopic behavior reflecting the interesting behavior of the corresponding continuum models~\cite{BJ87,KM92,DM95,B03}.
Due to the expected complexity of the material behavior, we only consider a simplified ``$(1+\eps)$-dimensional'' model.
In a sense, the model we investigate in this study is an intermediate one.
On one hand, it is more involved than a purely one-dimensional model as we consider two-dimensional deformations.
On the other hand, it is not fully two-dimensional since we restrict our attention to laminates -- i.e. 1D atomic chains. 
Thus, the considered model does not include genuinely two-dimensional phenomena such as the formation of e.g. branched microstructures 
which are expected to form for a large class of boundary conditions, see e.g.~\cite{CO12,Ch13}.
However, already in our simplified setting we are confronted with phenomena 
which, from a mathematical point of view, differ from the analogous one-dimensional situations:
\begin{itemize}
\item In the crucial compactness statements which are necessary in order to pass to the first order $\Gamma$-limit, one cannot argue via pure $L^{\infty}$ arguments. As the energy wells of the functional are not discrete an additional argument has to ensure compactness. For this we use the Friesecke-James-M\"uller $L^p$ rigidity theorem, c.f. \cite{FJM05}.
\item In order to estimate the density of defect points possessing high local energy
(needed for obtaining compactness and piecewise rigidity of the minimizing sequences),
 we apply a dimension separation approach, i.e. we first estimate
the number of high energy points for a few fixed horizontal atomic layers and then in the bulk between them.
This argument makes crucial use of the structure of the two wells (c.f. the proof of Proposition 3.1.).
\item Due to the lack of $L^{\infty}$ compactness and the prescribed deformation in the vertical direction, we develop slight modifications 
of Braides' and Cicalese's ~\cite{BC07} original, one-dimensional strategy of deriving the respective first order $\Gamma$-limits.
At this point horizontal and vertical ``cutting procedures'' are introduced (see e.g. Remark 4.1) that preserve
the non-interpenetration condition of the modified deformations.
\end{itemize}

We finally conclude the introduction by commenting on the organization of the remainder of the article:
\begin{itemize}
\item In Section 2 we introduce a class of discrete two-well Hamiltonians with prescribed properties.
Under a special periodicity assumption on the atomic configuration in the vertical direction, we then reduce these Hamiltonians to functions on certain generating 1D atomic chains.
\item In Section 3 we show compactness and asymptotic rigidity of the minimizing sequences as well as
of sequences whose rescaled energy remains controlled in the continuum limit.
\item In Section 4 we, rigorously, derive the first order $\Gamma$-limit for the chain Hamiltonian and, by that, obtain the limiting form of the surface energy.
\item In Section 5 we provide results of a numerical simulation underscoring the analytical results. These indicate exponential asymptotic decay of the boundary layers between twin configurations. 
\item In Section 6 we discuss the results and give an outlook.
\end{itemize}

\section{Setting and Notation}

In the sequel we work on the following parallelogram $\Omega \subset\mathbb{R}^2$ and for $n\in\N$ also consider the associated lattice $\Omega_n$ on it. 
For $\la = \frac{1}{n}$, we set:
\begin{align*}
\Omega &:= \left\{ z \Big|\; z=\ s\begin{pmatrix} 1\\ 0\end{pmatrix}+ t\frac{1}{\sqrt{2}}\begin{pmatrix}-1\\ 1 \end{pmatrix}, \ s,t\in [-1,1]  \right\}, \\
\Omega_{n}&:=\Omega \cap [\la\Z]^2.
\end{align*}
With a slight abuse of notation, we denote the lateral boundaries of the parallelograms by $\partial_x\Omega$ and $\partial_x \Omega_n$.
For brevity of notation, we further define the rescaled parallelograms 
\bes
\Omega_{n}^r:=\left\{ z \Big|\; z=\ s\begin{pmatrix} 1\\ 0\end{pmatrix}+ t\frac{1}{\sqrt{2}}\begin{pmatrix}-1\\ 1 \end{pmatrix}, \ s,t\in [-n,n]  \right\} \cap \Z^2.
\ees
Moreover, it proves to be convenient to introduce the notation
\begin{equation*}
\Omega(x_1,x_2):=  \Big \{z \Big| \;z =\begin{pmatrix} x \\ 0 \end{pmatrix}
+\frac{y}{\sqrt{2}}\begin{pmatrix} -1 \\ 1 \end{pmatrix}, \ \ x\in (x_1,x_2), \ \ y\in[-1,\,1]\Big\}\cap \Omega,
\end{equation*}
for a parallelogram determined by any pair of points $x_1,x_2\in [-1,1]$.
By $\mA_n$ we denote the set of all
deformations $u:\,\Omega_{n}
\rightarrow \mathbb{R}^2$ of a finite, $n$-dependent number of atoms from their initial reference configuration 
$\Omega_n$ such that $u$ is an orientation-preserving, non-selfinterpenetrating deformation, i.e.
\begin{equation}
\begin{split}
\mA_n:= \ &\Big\{u: \Omega_n\go\mR^2\Big|\ \mathrm{det}(u(x_2)-u(x_1),u(x_3)-u(x_1))\ge 0\ \text{for all}\\
&\ \{x_1,x_2,x_3\}\subset \Omega_n \quad
 \text{such that}\ \mathrm{diam}(x_1,x_2,x_3)=\sqrt{2}\la \ \\
&\ \text{and} \ \mathrm{det}(x_2-x_1,x_3-x_1)\ge 0\Big\}.\lb{An}
\end{split}
\end{equation}
Below we will identify such deformations with their piecewise affine interpolations 
\begin{equation*}
\begin{split}
\tilde{\mA}_n:=&\Big\{u: \Omega\go\mR^2:\,u\in C(\mR^2),\,u(x)\quad\text{is affine in}\quad\Omega_{ij}^\pm\Big|\\
&\ \ \mathrm{det}(u(x_2)-u(x_1),u(x_3)-u(x_1))\ge 0\\
&\ \ \text{for all}\ \{x_1,x_2,x_3\} \subset \Omega_n\ \text{such that}\ \mathrm{diam}(x_1,x_2,x_3)=\sqrt{2}\la\Big\},
\end{split}
\end{equation*}
where we define $\Omega_{ij}^\pm$ to be the triangles with the vertexes
\bes
\begin{pmatrix} i\la\\ j\la \end{pmatrix},\,\begin{pmatrix} (i+1)\la\\ j\la \end{pmatrix},\begin{pmatrix} i\la\\(j+1)\la\end{pmatrix}\quad
\text{and}\quad\begin{pmatrix} i\la\\ j\la \end{pmatrix},\,\begin{pmatrix} (i-1)\la\\ j\la \end{pmatrix},\begin{pmatrix} i\la\\(j-1)\la\end{pmatrix},
\ees
respectively. With a slight abuse of notation we will often identify $\mA_{n}$ and $\tilde{\mA}_{n}$ and omit the tildes in the notation in the sequel.

Moreover, in the remainder of the article we frequently make use of the notations
$f\lesssim\phi$ and $f\sim\phi$ in order to indicate the existence of positive, universal constants 
$c$ and $c_1,c_2$ such that the inequalities $f(\cdot)\le c\phi(\cdot)$ and
$c_1\phi(\cdot)\le f(\cdot)\le c_2\phi(\cdot)$ hold uniformly
in the set in which the arguments and parameters of the (positive) functions $f$ and $\phi$ are assumed to vary. \\

In the sequel, we will deal with two-dimensional Hamiltonians satisfying the following conditions:
\begin{itemize}
\item[(H1)] $H_n(u) =\sum\limits_{i,j=-n}^n \lambda_{n}^2 h\left(\frac{u^{ij}-u^{i\pm 1j}}{\lambda_{n}},\frac{u^{ij}-u^{ij \pm 1}}{\lambda_{n}}\right)$, 
where $u^{ij}:=u(i\la ,\,j\la )$. In this context, the $\pm$-signs denote
that $h$ depends on both the quantities with the $-$ and $+$ signs.
\item[(H2)] $h$ is rotation invariant, 
\item[(H3)] $h \in C^1$ has a super-linear, polynomial growth and satisfies 
\begin{align*}
h \left(\frac{u^{ij}-u^{i\pm 1j}}{\lambda_{n}},\frac{u^{ij}-u^{ij \pm 1}}{\lambda_{n}}\right)
\gtrsim \dist\left(\nabla u^{ij},
SO(2)U_{0}\cup SO(2)U_{1}\right)^p
\end{align*}
for some $p\in(1,\infty)$. Here, $\nabla u^{ij}$ is used as an abbreviation for the restriction of $\nabla u$
to the domain $\Omega_{ij}$, which is the union of the four triangles having one common vertex $[i\la,j\la]^T$. The inequality is assumed to hold uniformly in $\Omega_{ij}$.
\item[(H4)] The zero level set of the density $h$ is prescribed: On any domain $\Omega_{ij}$ the equation
$$h\left(\frac{u^{ij}-u^{i\pm 1j}}{\lambda_{n}},\frac{u^{ij}-u^{ij \pm 1}}{\lambda_{n}}\right)=0$$ 
is equivalent to $u = Q_{0}U_{0}x+ c_{0}$ or $u = Q_{1}U_{1}x+ c_{1}$.
Here $U_{0}, U_{1}$ are rank-one connected matrices such that for each matrix $U\in SO(2)U_i$, $i\in \{0,1\}$, there exist exactly 
two rank-one connected matrices in the respective other well,
$Q_{0}, Q_{1}\in SO(2)$ are arbitrary rotations and $c_0, c_1\in \mathbb{R}^2$ are constant off-set vectors. We further assume that $\det(U_0)=\det(U_1)=1$. 
\end{itemize}
The Hamiltonians satisfying the properties (H1)-(H4) are aimed at modeling a martensitic square-to-rectangular transformation in  $\mathbb{R}^2$ (which is a direct analog of cubic-to-tetragonal transformations in $\mathbb{R}^3$). 
One can easily show that the property (H3) implies that for all sufficiently small $\eta$
\begin{equation}
\begin{split}
h\left(\frac{u^{ij}-u^{i\pm 1j}}{\lambda_{n}},\frac{u^{ij}-u^{ij \pm 1}}{\lambda_{n}}\right)\le
\eta\Rightarrow& \left(\dist(\nabla u^{ij},SO(2)U_{0})\lesssim\eta^{1/p}\right.\\
& \left. \mbox{ or } \dist(\nabla u^{ij},SO(2)U_{1}) \lesssim\eta^{1/p}\right). 
\lb{H3r}
\end{split}
\end{equation}
In particular, estimate \rf{H3r} holds uniformly in $\Omega_{ij}$.\\

As an example of such an Hamiltonian we have the following atomistic two-well energy, $H_n(u)$, in mind:
\begin{eqnarray}
H_n(u)&:=&\sum\limits_{i,j=-n}^n \lambda_{n}^2 h\left(\frac{u^{ij}-u^{i\pm 1j}}{\lambda_{n}},\frac{u^{ij}-u^{ij \pm 1}}{\lambda_{n}}\right)\nonumber\\
&=&\sum_{i,j=-n}^{n}\lambda_{n}^{2}\left[\left(\left(\frac{u^{ij\pm 1}-u^{ij}}{\la}\right)^2- a^2\right)^2+\left(\left(\frac{u^{i\pm 1j}-u^{ij}}{\la}\right)^2- b^2\right)^2 \right. \nonumber\\
&&\left.+\left(\left(\frac{u^{ij\pm 1}-u^{ij}}{\la}\right)\cdot\left(\frac{u^{i\pm 1j}-u^{ij}}{\la}\right)\right)^2\right]\times\nonumber\\
&\times&\left[\left(\left(\frac{u^{ij\pm 1}-u^{ij}}{\la}\right)^2-  b^2\right)^2+\left(\left(\frac{u^{i\pm 1j}-u^{ij}}{\la}\right)^2- a^2\right)^2+ \right.
\nonumber\\ && \left.+\left(\left(\frac{u^{ij\pm 1}-u^{ij}}{\la}\right)\cdot\left(\frac{u^{i\pm 1j}-u^{ij}}{\la}\right)\right)^2\right ],
\label{HD}
\end{eqnarray}
where the parameters $a,\,b\in\mathbb{R}^+, a\neq b$, are chosen such that $ab=1$ (this corresponds to volume preserving transformations). 
In the above definition and below, we use a summation agreement: the sign $\pm$ in a term indicates that the latter should be replaced by the sum of the terms with all possible sign combinations, e.g.
\begin{eqnarray*}
\left((u^{ij\pm 1}-u^{ij})^2-(\la  a)^2\right)^2 &:= &\left((u^{ij + 1}-u^{ij})^2-(\la  a)^2\right)^2 \\
&& +\left((u^{ij - 1}-u^{ij})^2-(\la  a)^2\right)^2,\\
(u^{ij\pm 1}-u^{ij})\cdot(u^{i\pm 1j}-u^{ij})&:=&(u^{ij-1}-u^{ij})\cdot(u^{i-1j}-u^{ij})\\
&&+(u^{ij-1}-u^{ij})\cdot(u^{i+1j}-u^{ij})\\
&&+(u^{ij+1}-u^{ij})\cdot(u^{i-1j}-u^{ij})\\
&&+(u^{ij+1}-u^{ij})\cdot(u^{i+1j}-u^{ij}).
\end{eqnarray*}
We remark that, in particular, our functional (\ref{HD}) satisfies a condition similar to (H3):
\begin{align*}
H_n(u) \gtrsim \min\{\dist(\nabla u,SO(2)U_{0}\cap SO(2)U_1)^2, \dist(\nabla u,SO(2)U_{0}\cap SO(2)U_1)^4\}. 
\end{align*}
As will become evident from our proof of Theorem \ref{T1}, we are mainly
interested in the behavior of the Hamiltonian on a bounded set in gradient
space. Hence, for this Hamiltonian the lower bound effectively turns into a
quartic estimate (with $p=4$) with respect to the distance function. 
Thus, our special Hamiltonian (\ref{HD}) essentially satisfies the growth bounds required for the class of Hamiltonians defined via (H1)-(H4).\\
Moreover, the zeros of the first and second square brackets in (\ref{HD}) 
are given by all possible rotations of two rank-one connected affine deformations that are produced by the transformation matrices
\begin{equation}
U_0:=\left[\!\!\begin{array}{cc}
a&0\\
0&b\end{array}\!\!\right]\ \text{and}\ \ U_1:=\left[\!\!\begin{array}{cc}
b&0\\
0&a\end{array}\!\!\right], \nonumber
\end{equation}
respectively. Each matrix within one of the wells $SO(2)U_{0}$ or $SO(2)U_{1}$ is connected via two rank-one connections with the respective other well: 
There exist $Q, \tilde{Q}\in SO(2)$ such that 
\be
U_0-QU_1=\sqrt{2}\frac{a^2-b^2}{a^2+b^2}\begin{pmatrix} a\\-b \end{pmatrix}\otimes\frac{1}{\sqrt{2}}\begin{pmatrix}1\\ 1\end{pmatrix},\ \ 
U_0-\tilde{Q}U_1=\sqrt{2}\frac{a^2-b^2}{a^2+b^2}\begin{pmatrix} a\\b \end{pmatrix}\otimes\frac{1}{\sqrt{2}}\begin{pmatrix}1\\ -1\end{pmatrix}.
\lb{ROC}
\ee
Thus, it is possible for the material to form twins along these normals. 
We remark that for a general Hamiltonian satisfying (H1)-(H4) there is no restriction to assume that $U_0$ and $U_1$
are of the described form as an appropriate transformation reduces the general situation to this case. In the sequel we concentrate on this setting.\\  

Motivated by the structure of the wells and the example (\ref{HD}), we further restrict the class of deformations which we study.
For $\tau:=\begin{pmatrix}-a\\ b\end{pmatrix}$ we consider in this paper an additional constraint $u\in \mA_{n,\tau}$, where
\bea
\mA_{n,\tau}&:=&\Big\{u\in\mA_n\Big|\ u^{i+1j}-u^{ij+1}=-\la\tau^{i+j+1}\ \text{for all}\ (i,j)\in \Omega_{n}\Big\}\nonumber\\
&&\quad\text{where}\quad\tau^i\in SO(2)\tau\quad\text{for all}\quad i\in [-n,\,n].
\lb{Ant}
\eea
This implies that $u_n$ is represented via a 1D atomic chain on which the $i$-th atom on the base layer $j=0$ is  (vertically)
extended $n$  (and $-n$) times in the direction of the corresponding vector $\la\tau^i$, 
which depends on the horizontal position $i$. For appropriate boundary conditions (see details below), this is a reasonable assumption as in this case one expects the ground states of any Hamiltonian satisfying (H1)-(H4) to stay locally close to 
laminar configurations formed by pairs of the two martensitic variants.
Note that the particular case of an atomic chain extended uniformly by the vector
$\tau$ in the vertical direction is included in the definition of $ \mA_{n,\tau}$. \\

\begin{figure}[t]
\centering
\includegraphics[scale=0.8]{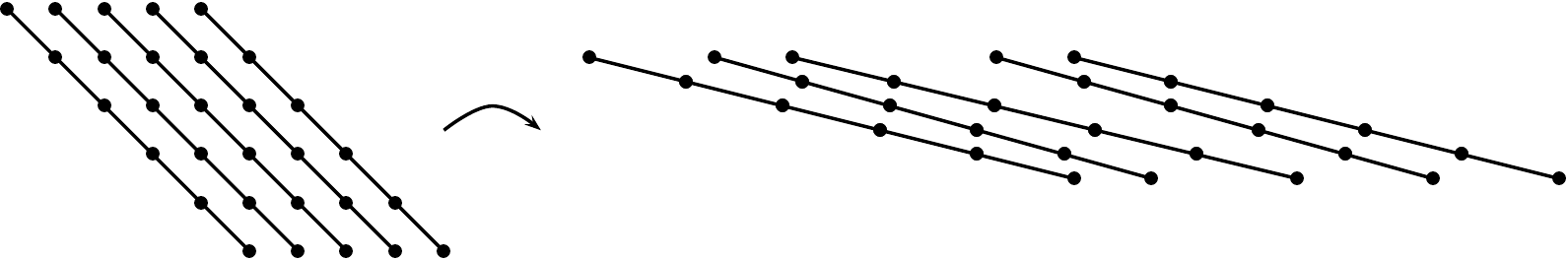}
\caption{An example of a constrained deformation.}
\end{figure}

Returning to our model case, the restriction $u\in \mA_{n,\tau}$ allows us to reduce (\ref{HD}) to a function on the generating
1D chain in the parallelogram $\Omega_n$. More precisely, in this case
\begin{equation}
\begin{split}
H_n(u)=&\lambda_{n}^2  \sum_{i,j=-n}^{n}h\left(\frac{u^{i}_n - u^{i\pm 1}_n}{\la}, \tau^i_n, \tau_n^{i\pm 1}, j\right)\\ 
:=& \lambda_{n}^{2}  \sum_{i,j=-n}^{n}\left [\left(\left(\frac{u^{i\pm 1}-u^{i}}{\la}+(j\pm 1)\delta_n^{i\pm 1}\pm\tau_n^i\right)^2-  a^2\right)^2+\right.\\ 
& \ \ \ +\left(\left(\frac{u^{i\pm 1}-u^{i}}{\la}+j\delta_n^{i\pm 1}\right)^2- b^2\right)^2\\
& \ \ \ +\left.\left(\left(\frac{u^{i\pm 1}-u^i}{\la}+(j\pm 1)\delta_n^{i\pm 1}\pm\tau_n^i \right)\cdot\left(\frac{u^{i\pm 1}-u^{i}}{\la}+j\delta_n^{i\pm 1}\right)\right)^2\right]\\ 
& \ \times \left[\left(\left(\frac{u^{i\pm 1}-u^{i}}{\la}+(j\pm 1)\delta_n^{i\pm 1}\pm\tau_n^i\right)^2-  b^2\right)^2 \right.\\
& \ \ \ +\left(\left(\frac{u^{i\pm 1}-u^{i}}{\la}+j\delta_n^{i\pm 1}\right)^2-  a^2\right)^2\\
& \ \ \ \left.+\left(\left(\frac{u^{i\pm 1}-u^i}{\la}+(j\pm 1)\delta_n^{i\pm 1}\pm\tau_n^i\right)\cdot\left(\frac{u^{i\pm 1}-u^{i}}{\la}+j\delta_n^{i\pm 1}\right)\right)^2\right],
\label{DC}
\end{split}
\end{equation}
where we denoted the atoms of the generating 1D chain by  $u^i:=u(i\la,0)$ for
$i\in [-n,n]\cap\Z$ and the corresponding discrepancy between neighboring shift vectors
by $\delta^{i\pm 1}_n:=\tau^{i\pm 1}_n-\tau^i_n$. 
We remark that for an arbitrary Hamiltonian satisfying (H1)-(H4) the reduction to atomic chains, i.e. $u\in \mA_{n,\tau}$, follows analogously. Moreover, we point out that, in passing to the atomic chains, we also restrict the underlying (deformed) domain to the previously introduced parallelograms $\Omega$ and $\Omega_n$.\\

In this paper we investigate global minimizers of the reduced Hamiltonian (\ref{DC})
among all deformations $u\in \mA_{n,\tau}$ satisfying  {\it Dirichlet boundary conditions} prescribed by a certain linear deformation having gradient $F\in\mR^{2\times2}$ with $\mathrm{det}(F)=1$.  More precisely, we assume that $u$ may
be extended to the whole lattice $\mathbb{Z} \times [-n,n]$ such that for the generating chain it holds
\bes
u^i=F\left(\begin{array}{c}i\la\\0\end{array}\right)\ \text{and}\ \tau^i\equiv\tau\ \text{if}\  i\le -n\ \ \text{and}\ \ \text{if}\ i\ge n.
\ees
As we are mainly interested in the emergence of surface energy contributions, we do not consider the full class of possible boundary conditions. Instead, we restrict our attention to the case of linear boundary data leading to zero bulk energy contributions in the continuum limit. Applying the results of~\cite{AC04, BS12}
one can prove the existence of the zero order $\Gamma$-limit for $H_{n}$. Due to these results on the derivation of continuum limits, the zero set of the continuum elastic energies, on the one hand, contain at least $(SO(2)U_0\cup SO(2)U_1)^{qc}$ -- the quasiconvexification of the wells --
as the resulting continuum limits are determined by a non-negative, quasiconvex energy density.  

On the other hand, as $\dist(x,K)\geq \dist(x,K^{qc})$ for each arbitrary set $K\subset \mR^{2\times 2}$, one deduces
\begin{eqnarray}
H_{n}(u_n) &&= \sum\limits_{i,j=-n}^{n}\lambda_n^2  h\left(\frac{u^{ij}-u^{i\pm 1j}}{\lambda_{n}},\frac{u^{ij}-u^{ij \pm 1}}{\lambda_{n}}\right)\nonumber\\
&& \gtrsim\sum\limits_{i,j=-n}^{n} \lambda_{n}^2 \dist(\nabla u_n^{ij}, SO(2)U_{0}\cup SO(2)U_{1})^p\nonumber\\
&& \gtrsim \int\limits_{\Omega}\dist(\nabla u_n,  SO(2)U_{0}\cup SO(2)U_{1})^pdx\nonumber\\
&& \gtrsim \int\limits_{\Omega}\dist(\nabla u_n,  (SO(2)U_{0}\cup SO(2)U_{1})^{qc})^pdx,
\lb{Es1}
\end{eqnarray}
as a result of (H3) and (H4).
The estimate \rf{Es1} then implies that the  zero set of the zero order  $\Gamma$-limit  is exactly given by $(SO(2)U_{0}\cup SO(2)U_{1})^{qc}$.

Within $(SO(2)U_0\cup SO(2)U_1)^{qc}$ affine boundary conditions inducing twin configurations with zero bulk energy contributions and
prescribed chain direction, $\tau$, are associated with deformation gradients of the form
\begin{align}
\label{eq:newBC}
F_{\lambda} = (1-\lambda) U_{0} + \lambda QU_{1}, \;\; \lambda\in [0,1],
\end{align}
where $Q\in SO(2)$ corresponds to the rotation from \rf{ROC} and
\begin{align*}
F_{\lambda}\begin{pmatrix} 1\\-1 \end{pmatrix} = \begin{pmatrix} a \\-b \end{pmatrix} = -\tau.
\end{align*} 
Thus, it turns out to be convenient to introduce a subspace of $\mathcal{A}_{n,\tau}$ 
which incorporates these (Dirichlet) data into our class of functions:
For $F_{\lambda}\in \mathbb{R}^{2\times 2}$ as above, we define
\begin{align}
\label{eq:BData}
\mA_{n,\tau}^{F_{\lambda}}  := \mA_{n,\tau}\cap \{u: \Omega_n \rightarrow \mathbb{R}^2 \big| u = F_{\lambda}x\mbox{ on } \partial_x\Omega_n\}
\end{align}
In the sequel, we investigate the
limiting behavior of minimizers ($n \rightarrow \infty$) in the class (\ref{eq:BData})
as well as the emergence of surface energy contributions.

\section{Rigidity of Minimizers and Limiting Form for the Surface Energy}

Our first main theorem shows that minimizing sequences to (\ref{DC})
considered with boundary conditions prescribed by (\ref{eq:newBC}) converge to piecewise affine deformations.
On each of the continuity subintervals of its gradient, the respective deformation coincides with a rotation of one of the two martensitic variants, 
i.e. it corresponds to one of the transformations in $SO(2)U_0\cup SO(2)U_1$. The rotations occurring in the rigidity result are not arbitrary: The gradients of the deformation have to satisfy a rank-one condition along the $\begin{pmatrix}1\\ 1 \end{pmatrix}$ normal direction and therefore the rotations have to coincide either with $\mathrm{Id}$ or $Q$. Although our statements are, for convenience, formulated for the Hamiltonian (\ref{HD}), our arguments do not use the specific properties of this Hamiltonian. Hence, the results remain true for the respective 1D atomic chains corresponding to any Hamiltonian satisfying the conditions (H1)-(H4).
\begin{theorem}
Let $F_{\lambda}$, $\lambda\in [0,1]$, be as above. Let  $\{u_n\}_{n \in \N} \subset\mA_{n,\tau}^{F_{\lambda}}$ be a sequence of
minimizers of (\ref{DC}). Then there exists a
number $K\in\N$ such that (for a not-relabeled subsequence)
\begin{itemize}
\item[(i)] $u_n\rightarrow u\ \text{in}\ W^{1,4}(\Omega,\mathbb{R}^2)$,
\item[(ii)] for each $s\in \{1,...,K-1\}$ there exist $m_s\in\{0,1\}$, $x_s\in [-1,1]$ such that
\begin{equation*}
\triangledown u(z) \equiv Q^{m_{s}}U_{m_s}
\end{equation*}
$\text{ for } z \in \Omega(x_s,x_{s+1})$
where $Q^0:=\mathrm{Id},\quad Q^1:=Q$ and $x_1=-1,\,x_K=1$.
\item[(iii)] 
$\bigcup\limits_{s=1 }^{K-1}[x_{s},x_{s+1}] = [-1,1].$
\end{itemize}
\label{T1}
\end{theorem}

\begin{remark}
The choice of $p=4$ in (i) is arbitrary; in fact our proof shows that it is possible to deduce $u_n \rightarrow u$ in $W^{1,p}$ for any $p\in (1,\infty)$.
\end{remark}

\begin{remark}
Before proceeding with the proof, we comment on its structure. As in most similar proofs, we first construct a comparison function in order to obtain an upper bound on the energy. We then crucially exploit the two-well structure of the Hamiltonian and the geometry of the rank-one connections between the wells (c.f. Step 2). Here the key observation is that in matrix space any rank-one line between the wells only intersects the respective other well once. This allows to extend the control from certain particular horizontal layers $j_{-1}^n, j_{0}^n, j_{1}^n$ to the whole vertical stripe between them, c.f. the calculations following \rf{ContAs}.
This information then allows to apply the Friesecke-James-M\"uller rigidity theorem \cite{FJM05}.
\end{remark}

\begin{proof}
\textit{Step 1: Constructing an appropriate comparison function.} Below $c>0$ denotes a constant that may vary from line to line but does not depend on any
parameter of the problem. We first consider the following piecewise affine comparison function:
\begin{equation}
\label{eq:comp2}
\begin{split}
u(z) = 
\left\{
\begin{array}{ll}
F_{\lambda}z  & \mbox{for } z \in\Big \{\begin{pmatrix} x \\ 0 \end{pmatrix} + \frac{s}{\sqrt{2}}\begin{pmatrix} -1 \\ 1 \end{pmatrix},\,x\in (-\infty,-1],\,s\in \Rl \Big \}\cap \Omega,\\ 
U_{0}z + c_{1}, & \mbox{for } z \in \Big \{\begin{pmatrix} x \\ 0 \end{pmatrix} + \frac{s}{\sqrt{2}}\begin{pmatrix} -1 \\ 1 \end{pmatrix}, \,x\in (-1,1-2\lambda),\,  s\in \Rl \Big \}\cap \Omega,\\
QU_{1}z + c_{2}, & \mbox{for } z \in \Big \{\begin{pmatrix} x \\ 0 \end{pmatrix} + \frac{s}{\sqrt{2}}\begin{pmatrix} -1 \\ 1 \end{pmatrix},\,x\in (1-2\lambda,1),\,   s\in \Rl\Big \}\cap \Omega,\\
F_{\lambda}z & \mbox{for } z\in \Big \{\begin{pmatrix} x \\ 0 \end{pmatrix} + \frac{s}{\sqrt{2}}\begin{pmatrix} -1 \\ 1 \end{pmatrix},\,x\in [1,\infty),\,s\in \Rl \Big \}\cap \Omega,
\end{array} 
\right.
\end{split}
\end{equation}
where the constants $c_{1}, c_{2} \in \mathbb{R}^2$ are chosen such that the resulting function $u(z)$ is continuous (which is possible due to the rank-one connections between the wells). 
Let $\{ u_n \}_{n\in\N}$ be a minimizing sequence corresponding to the energies $H_n(\cdot)$ considered with the boundary condition \rf{eq:newBC}. Then 
\begin{align}
\label{eq:comp}
H_n(u_n) \leq H_n(u) \leq  c\la.
\end{align}
We consider the rescaled Hamiltonian (which corresponds to surface energy contributions originating from the boundaries and interfaces):
\begin{equation}
H_n^1(u_n):= \frac{H_n(u_n)}{\la}.
\lb{SE}
\end{equation}
From (\ref{eq:comp}) we observe
\begin{align}
\label{eq:comp11}
\la\sum_{i,j=-n}^{n}h\left(\frac{u^{i}_n - u^{i\pm 1}_n}{\la},\tau^i_n,\,\tau_n^{i\pm 1}, j\right)\leq c.
\end{align}

\textit{Step 2: Estimating the number of the jumps between the wells.}
Let $0<\alpha<1$. Then the energy bound \rf{eq:comp11} yields control on the energy per horizontal line and on the energy per lattice point:
\sbea
&&\#\left\{j:\ \la\sum_{i=-n}^{n}h\left(\frac{u^{i}_n - u^{i\pm 1}_n}{\la},\tau^i_n,\,\tau_n^{i\pm 1}, j\right)\ge n^{-\alpha}\right\}\lesssim n^{\alpha}, \label{stratN}\\
&&\#\left\{(i,j):\ h\left(\frac{u^{i}_n - u^{i\pm 1}_n}{\la},\tau^i_n,\,\tau_n^{i\pm 1}, j\right)\ge n^{-\alpha}\right\}\lesssim n^{1+\alpha}.\label{eq:F}
\seea
In particular, for any sufficiently small $\delta>0$, \rf{stratN} and \rf{eq:F} imply 
that for sufficiently large $n\in\N$ there exist a large number of horizontal lines with good energy estimates, i.e. it is possible to find $j_{-1}^n,j_0^n,j_1^n$ such that
\bea
\lb{jD}
&&j_{-1}^n\in[-n,-n+ 2\delta n],\ j_0^n\in[-\delta n,\delta n],\ j_{1}^n\in[n-2\delta n,n],\\
&&\la\sum_{i=-n}^{n}h\left(\frac{u^{i}_n - u^{i\pm 1}_n}{\la},\tau^i_n,\,\tau_n^{i\pm 1}, j_l^n\right)\lesssim n^{-\alpha},\label{eq:good1}\\
&&\#\left\{i\in [-n,n]:\ h\left(\frac{u^{i}_n - u^{i\pm 1}_n}{\la},\tau^i_n,\,\tau_n^{i\pm 1}, j_l^n\right)\ge
 n^{-\alpha}\right\}\lesssim\delta^{-1}n^{\alpha},
\lb{BP1}
\eea 

for $l\in\{-1,0,1\}$.
This is a consequence of the following observations:
\begin{itemize}
\item Due to \rf{stratN} and $\alpha<1$, the number of $j\in [-n,n]$ violating \rf{eq:good1} is smaller than $\frac{\delta}{2}n$ if $n$ is sufficiently large.
\item If \rf{BP1} were wrong, for example, for all $j\in [-n,-n+2\delta n]$, this would entail that for any constant $c>0$
\begin{equation*}
\#\left\{(i,j):\ h\left(\frac{u^{i}_n - u^{i\pm 1}_n}{\la},\tau^i_n,\,\tau_n^{i\pm 1}, j\right)\ge n^{-\alpha}\right\} >\delta^{-1}cn^{\alpha}2\delta n
\end{equation*}
which contradicts \rf{eq:F}.
\end{itemize}
Setting
$ \tilde{c}=\left(\frac{b^2-a^2}{100(a^2+b^2)}\right)^4$ (note that in the general case of a Hamiltonian with $p$-growth a similar choice can be made),
we observe that we may further choose $j_{-1}^n,j_0^n,j_1^n$ such that there exists a number $M_{\delta}>0$, independent of $n$ with
\be
\#\left\{i\in [-n,n]:\ h\left(\frac{u^{i}_n - u^{i\pm 1}_n}{\la},\tau^i_n,\,\tau_n^{i\pm 1}, j_l^n \right)\ge\tilde{c}\right\}\le M_{\delta}\quad\text{for}\quad l\in\{-1,0,1\}.
\lb{BP2}
\ee
Indeed, if this were false for all lines in the intervals given by (\ref{jD}), this would imply that for any $M >0$ and all sufficiently large $n$ it holds:
\begin{align*}
\la\sum_{i,j=-n}^{n}h\left(\frac{u^{i}_n - u^{i\pm 1}_n}{\la},\tau^i_n,\,\tau_n^{i\pm 1}, j\right)\geq  \frac{1}{n} M \delta n = M \delta.
\end{align*}
Taking $M$ such that $M \delta>c$, would then contradict \rf{eq:comp11}. 
Therefore, choosing $M_{\delta}$ sufficiently large, the fraction of lines satisfying (\ref{BP2}) becomes sufficiently large in order to
find horizontal lines satisfying \rf{jD}, \rf{eq:good1}, \rf{BP1} and \rf{BP2} simultaneously. Moreover, by density considerations similar 
to the previous ones and by potentially enlarging the constants in \rf{jD}, \rf{eq:good1}, \rf{BP1} and \rf{BP2} by a factor $\sim \frac{100}{\delta}$, 
we may additionally with out lost of generality assume that $j^{n}_1-j^n_0 = j^n_0 - j^{n}_{-1}$. This follows from the facts that
\begin{itemize}
\item equation (\ref{eq:good1}) is satisfied by $n- cn^{\alpha}$ choices of $j\in [-n,n]$,
\item equations (\ref{BP1}) and \rf{BP2} hold for $n-\epsilon n$ different choices of $j\in [-n,n]$,
\end{itemize}
and the observation that the constant $\epsilon$ can be made arbitrarily small by choosing the constants in the respective estimates sufficiently large. We stress that the points at which (\ref{BP2}) holds are the only places along the horizontal lines $\{j_{-1}^n,j_0^n,j_1^n\}$ at which jumps \emph{between} the wells can occur (later we will see that these points and their vertical extensions are in fact \emph{globally} the only points at which such large jumps may occur).
Due to the uniformity in $n$, estimate (\ref{BP2}) will play a crucial role in controlling the location and number of the large jumps, c.f. step 3.\\
As a last immediate consequence of the upper estimate  on the energy (\ref{eq:comp11}),
we deduce an $L^{\infty}$ bound on minimizing configurations which is uniform in $n$: Indeed, (\ref{eq:comp11}) directly implies
\begin{align*}
\lambda_n \sum\limits_{j=-n}^n h\left(\frac{u^{i}_n - u^{i\pm 1}_n}{\la},\tau^i_n,\,\tau_n^{i\pm 1}, j\right) \leq c.
\end{align*}
Choosing a constant $C_1$ of the size $C_1 \geq 100 c$, allows to conclude that for a fixed 
$i\in[-n,n]$ not more than one percent of all points $(i,j)$ have a local energy exceeding $C_1$.
Hence, for each $i\in (-n,n)$ it is possible to find vertical atoms of distance $\sim 2n$
such that their local energy, $h\left(\frac{u^{i}_n - u^{i\pm 1}_n}{\la},\tau^i_n,\,\tau_n^{i\pm 1}, j\right)$,
is less than $C_1$. Due to the chain structure of the Hamiltonian, the local energy of any vertical point $j$ is then bounded by $3C_1$. As a consequence, we obtain
\begin{align}
\label{eq:Linf}
|\nabla u^{ij}_n| \lesssim c < \infty \mbox{ uniformly for } n\in \N \mbox{ and }i,\,j\in \{-n,...,n\}.
\end{align}

\begin{figure}[t]
\centering
\includegraphics{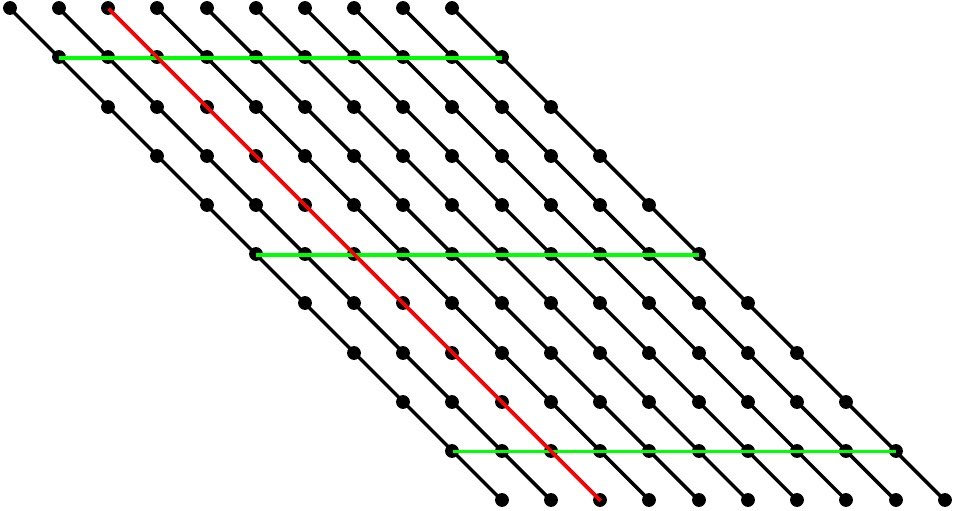}
\caption{In estimating the energy, we consider three ``good'' layers $j=_0^n,\,j_{\pm1}^n$ which are indicated in green. 
Along the vertical lines (e.g. along the red line) the 1D chain structure implies that the gradient changes along rank-one directions in the matrix space.}
\end{figure}

We proceed by considering the points having low energy and show that for a large fraction of points the deformation is already ``approximately laminar''.
For this, we define a position $i\in [-n,n]$ on the horizontal lines, to be ``simultaneously good'' for  $j_{-1}^n,j_0^n,j_1^n$  if
\bes
h\left(\frac{u^{i}_n - u^{i\pm
    1}_n}{\la},\tau^i_n,\,\tau_n^{i\pm 1},
j_l^n\right)\le  n^{-\alpha}\quad\text{holds for all}\quad l\in\{-1,0,1\}.
\ees
Note, that due to \rf{BP1} the number of $i\in [-n,n]$ which are not ``simultaneously good'' is $\lesssim\delta^{-1}n^{\alpha}$.
Property \rf{H3r} implies that for each ``simultaneously good'' position $i\in [-n,n]$, one has
\be
(\dist(\nabla u_n^{i-j_l^n j_l^n},SO(2)U_{0})\lesssim n^{-\alpha/4})\ \mbox{or}\ (\dist(\nabla u_n^{i-j_l^n j_l^n},SO(2)U_1)\lesssim n^{-\alpha/4})
\lb{Ms2}
\ee
for $l\in\{-1,0,1\}$.

Next, we claim that due to the two-well structure of our Hamiltonian a stronger property is satisfied:
For sufficiently large $n$ and each fixed ``simultaneously good'' $i$ exactly one of the inequalities in \rf{Ms2} holds
for \emph{all} $l\in\{-1,0,1\}$. This will then imply that it is impossible to switch between the wells along the vertical direction
if one starts at a ``simultaneously good'' position.  
In order to prove this claim, we argue by contradiction. Let us fix a ``simultaneously good''
$i\in [-n,n]$ and assume, for example, that the following holds
\be
(\dist(\nabla u_n^{i-j_0^n j_0^n},SO(2)U_{0}) \lesssim n^{-\alpha/4})\ \mbox{and}\ (\dist(\nabla u_n^{i-j_1^n j_1^n},SO(2)U_{1}) \lesssim n^{-\alpha/4}).
\lb{ContAs}
\ee
By definition \rf{DC} of our 1D chain, one has the following relations
\begin{equation}
\begin{split}
u_n^{i-j_1^n j_1^n\pm 1}-u_n^{i-j_1^n j_1^n}&=u_n^{i-j_0^n j_0^n\pm 1}-u_n^{i-j_0^n j_0^n}+(j_1^n-j_0^n)\la\delta_n^{i\pm 1},\\
u_n^{i-j_1^n\pm 1 j_1^n}-u_n^{i-j_1^n j_1^n}&=u_n^{i-j_0^n\pm 1 j_0^n}-u_n^{i-j_0^n j_0^n}+(j_1^n-j_0^n)\la\delta_n^{i\pm 1}.
\lb{Rel0}
\end{split}
\end{equation}
These, together with \rf{ContAs}, imply that if 
\bes
||\nabla u_n^{i-j_0^n j_0^n}- Q_{1,n}U_{0}||_{ C(\Omega_{i-j_0^nj_0^n})}\lesssim n^{-\alpha/4}\quad\text{for some}\quad  Q_{1,n} \in SO(2),
\ees
then necessarily there exists  $\gamma_n \in (-\pi/2,\pi/2)$ such that
\be
||\nabla u_n^{i-j_1^n j_1^n}- Q_{\gamma_n}Q_{1,n} U_{1}||_{C(\Omega_{i-j_1^nj_1^n})} \lesssim n^{-\alpha/4},
\lb{Rel1}
\ee
where $ Q_{\gamma_n}\in SO(2)$ denotes a rotation by the angle $\gamma_n$
and the following conditions on $(j_1^n-j_0^n)\delta_n^{i\pm 1}$ and $\gamma_n$ hold:
\begin{equation}
\label{Rel}
\begin{split}
\left|\sin(\gamma_n)-\left[\frac{a^2-b^2}{a^2+b^2}\right]\right|&\lesssim n^{- \alpha/4},\\
\left|(j_1^n-j_0^n) Q_{1,n}^{-1}\cdot\delta_n^{i\pm 1}+\begin{pmatrix}a\sin\gamma_n\\-b\sin\gamma_n\end{pmatrix}\right|&\lesssim n^{-\alpha/4}.
\end{split}
\end{equation}
The inequalities \rf{Rel1}--\rf{Rel} are connected to the fact that a deformation with gradient in
$SO(2)U_1$ can be obtained from one of the respectively associated rank-one connected matrices in $SO(2)U_0$ by a shear deformation. 

Next, let us suppose that 
\be
\dist(\nabla u_n^{i-j_{-1}^n j_{-1}^n},SO(2)U_{1})\lesssim n^{-\alpha/4}.
\lb{ContAs1}
\ee
Then, analogous to the derivation of \rf{Rel}, there exists an angle $\gamma_{1,n}$ such that
\bes
||\nabla u_n^{i-j_{-1}^n j_{-1}^n}-Q_{\gamma_{1,n}} Q_{ 1,n}U_1||_{C(\Omega_{i-j_{-1}^nj_{-1}^n})}\lesssim n^{-\alpha/4},
\ees
and the following relations hold between $\gamma_n$ and $\gamma_{1,n}$
\beas
\left|a\sin\gamma_{1,n}+a\sin\gamma_n\right|&\lesssim&n^{- \alpha/4},\nonumber\\
\left|2b-(a\cos\gamma_n+a\cos\gamma_{1,n})\right|&\lesssim&n^{- \alpha/4}.
\eeas
For sufficiently large $n$ and a correspondingly sufficiently small error $n^{-\alpha/4}$ the last two inequalities become incompatible with the
first one in \rf{Rel} as soon as $a\neq b$. Therefore, such an angle $\gamma_{1,n}$ cannot exist, which thus yields a contradiction to \rf{ContAs1}. Consequently,
\bes
\dist(\nabla u_n^{i-j_{-1}^n j_{-1}^n},SO(2)U_0)\lesssim n^{-\alpha/4}
\ees
holds. Proceeding analogous to the considerations in \rf{Rel0}-\rf{Rel1}
one concludes in this case that
\begin{align}
||\nabla u_n^{i-j_{-1}^n j_{-1}^n}-Q_{ 1,n}U_{0}||_{C(\Omega_{i-j_{-1}^nj_{-1}^n})}&\lesssim n^{-\alpha/4},\nonumber\\
\left|(j_0^n-j_{-1}^n)\delta_n^{i\pm 1}\right|& \lesssim n^{-\alpha/4}. \lb{Rel2}
\end{align} 
The last inequality, however, by \rf{Rel0} implies that
\bes
||\nabla u_n^{i-j_1^n j_1^n}-Q_{ 1,n}U_0||_{C(\Omega_{i-j_1^nj_1^n})}\lesssim n^{-\alpha/4}.
\ees
Therefore, one arrives at a contradiction to the starting assumption \rf{ContAs}. 
This and analogous considerations in the remaining cases, prove the claim for simultaneously good points. Finally, we conclude that the estimates can be extended along the whole vertical direction: 
For each ``simultaneously good'' $i\in [-n,n]$ there exists $ Q_{i,n}\in SO(2)$ such that either 
\begin{equation}
\begin{split}
||\nabla u_n^{i-j j}-Q_{i,n}U_{0}||_{C(\Omega_{i-jj})}&\lesssim n^{-\alpha/4}\quad\text{for all}\quad j\in [-n,n],  \\
\ \ \mbox{ or } ||\nabla u_n^{i-j j}-Q_{i,n}U_{1}||_{C(\Omega_{i-jj})}&\lesssim n^{-\alpha/4}\quad\text{for all}\quad j\in [-n,n].
\lb{Ms3}
\end{split}
\end{equation}
This follows from the previous argument and the second estimate in (\ref{Rel2}). Similar estimates can be shown under $p$-growth assumptions
on the two-well Hamiltonian.\\

\textit{Step 3: Jumps within the wells and the FJM rigidity theorem.}
As a consequence of \rf{Ms3} and (\ref{eq:Linf}) we have,
 $||u_{n}||_{W^{1,p}(\Omega)}\leq C$ for any $p\in (1,\infty)$. On passing to subsequences, we may therefore conclude that
\begin{itemize}
\item there exist $K\in\N$, $x_{1},...,x_{K}\in (-1,\,1)$ independent of $n$,
\item and for any $n$ there exist associated points
$x_{1}^{n},...,x_{K}^{n}\in (-1,\,1)$ and $y_{s,1}^{n},...,y_{s, K_s^n}^{n} \in (x_{s}^{n},x_{s+1}^{n})$ 
\end{itemize}
such that
\begin{itemize}
\item $ x_{s}^{n} \rightarrow x_{s}, \; s\in \{1,...,K\},$
\item $u_n \rightharpoonup u \mbox{ in } W^{1,4}(\Omega), $ (here the choice $p=4$ is arbitrary),
\item as a consequence of \rf{Ms2} the following dichotomy holds in the interval $(x_{s}^{n},x_{s+1}^{n})$: For each $i$ with
$\la i\in (y_{s,l}^{n},y_{s,l+1}^{n}) \subset (x_{s}^{n},x_{s+1}^{n})$ and $l\in\{1,...,K_s^n\}$, either
\begin{align}
\label{eq:wells}
& \dist(\nabla u_n^{i-j j}, SO(2)U_{0})\lesssim n^{-\alpha/4}\quad  \mbox{ or }
\dist(\nabla u_n^{i-j j}, SO(2)U_{1})\lesssim n^{-\alpha/4}
 \end{align}
for all $j\in [-n,n]$. 
\end{itemize}
Here the points $x_s$, $s\in \{1,...,K\}$, are obtained by taking a limit of an appropriate subsequence of the rescaled points $\frac{i}{n}$ from the set in (\ref{BP2}). Similarly, the points $y_{s,l}^n$ are obtained from the points which are not simultaneously good, i.e. at least one of the points on the vertical lines $j_{-1}^n,j_{0}^n$ or $j_{1}^n$ carries a local energy of less than $\tilde{c}$ but larger than $n^{-\alpha}$. Since the number of the latter points can possibly increase with growing $n$, there is no uniform a priori bound on it. 
However, our previous considerations on the simultaneously good points lying between these allows to deduce (\ref{eq:wells}).

Note, that if for a certain $l\in\{1,...,K_s^n\}$, a certain $\la i\in
(y_{s,l}^{n},y_{s,l+1}^{n}) \subset (x_{s}^{n},x_{s+1}^{n})$ 
and a certain $j\in [-n,n]$ one has $\dist(\nabla u_n^{i-j j}, SO(2)U_{0}) \leq   cn^{-\alpha/4}$, 
then this remains true for all $l\in \{1,...,K_s^n\}$, $\la i\in
(y_{s,l}^{n},y_{s,l+1}^{n})$  and  $j\in [-n,n]$. 
Indeed, firstly all atoms in $\Omega(x_{s}^n,x_{s+1}^n)$  lying on one of the horizontal slices
determined by $j_0^n,\,j_1^n,j_{-1}^n$ are in a
neighborhood of a single, common well, say $SO(2)U_{0}$. If this were not true, then
there would necessarily exist at least two internal neighboring points
$(i,j)$ and $(i+1,j)$ for some $j\in\{j_0^n,\,j_1^n,j_{-1}^n\}$ 
such that $\nabla u^{ij}$ and $\nabla u^{i+1j}$ belong to disjoint neighborhoods of the two wells $SO(2)U_0$ and $SO(2)U_1$. 
However, such points cannot exist as the atoms $(i,j)$ and $(i+1,j)$ have a
common edge and, by definition, have local energies smaller than $\tilde{c}$. In particular, there cannot be a jump of the required size at these points. 
Secondly, by the argument used in the Step 2 all atoms lying on the vertical
extension of simultaneously good points on the slices
$j_0^n,\,j_1^n,j_{-1}^n$ belong to the same well $SO(2)U_{0}$. By definition of $y_{s,l}^{n}$ the set of these
atoms  coincides with the set of ones having  $\la i\in
(y_{s,l}^{n},y_{s,l+1}^{n}) \subset (x_{s}^{n},x_{s+1}^{n})$  with
$l\in\{1,...,K_s^n\}$ and $j\in [-n,n]$ and the above statement follows.\\

\textit{Step 3a: Non-degenerate intervals.}
We distinguish two cases: In the first one we consider $s\in \{1,...,K\}$
such that $x_{s}\neq x_{s+1}$, 
in the second case the jump planes are allowed to collapse in the limit. 
We start discussing the first alternative. For a fixed $s\in \{1,...,K\}$ there exists $m_s\in\{0,1\}$ such that (\ref{eq:wells}) holds with a single matrix, $U_{m_s}$, for any $i$ with $\la i\in (y_{s,l}^{n},y_{s,l+1}^{n}) \subset (x_{s}^{n},x_{s+1}^{n})$, $l\in\{1,...,K_s\}$
(as there are no jumps between the different wells within $(x_{s}^{n},x_{s+1}^{n})$). Thus, we obtain 
\bes
\int\limits_{\Omega(x_{s}^n,x_{s+1}^n)}\dist(\nabla u_n,SO(2)U_{m_s})^4 dz \lesssim n^{-\alpha} + \frac{n^{\alpha}}{n}, 
\ees
as, by \rf{stratN}, the number of ``bad'' vertical stripes is controlled by $n^{\alpha}$, the
horizontal length of each stripe is given by $n^{-1}$ and its energy is bounded by a constant. Therefore, 
\bes
\int\limits_{\Omega(x_{s}^n,x_{s+1}^n)}\dist(\nabla u_n,SO(2)U_{m_s})^4 dz \rightarrow 0 \mbox{ as } n\rightarrow \infty.
\ees
Next, we apply the (non-linear) quantified Liouville ($L^p$-rigidity) theorem
of Friesecke-James-M\"uller~\cite{FJM05}: 
For each $n\in \N$ and each interval $(x_{s}^{n},x_{s+1}^{n})$, $s\in \{1,...,K-1\}$, there exist $V^{n}_s\in SO(2)U_{m_s}$ such that
\begin{align*}
\int\limits_{\Omega(x_{s}^{n},x_{s+1}^{n})}|\nabla u_n - V^{n}_s|^4\,dz \leq 
c \int\limits_{\Omega(x_{s}^{n},x_{s+1}^{n})}\dist(\nabla u_{n},SO(2)U_{m_s})^4 dz. 
\end{align*}
Note that the constant $c$ in the last estimate can be chosen uniformly on the domains $\Omega(x_{s}^{n},x_{s+1}^{n})$ if $x_s\ne x_{s+1}$.
Since $SO(2)U_{m_s}$ is compact, for each $s\in \{1,..., K\}$ there exists a subsequence $V^{n}_s\rightarrow V_s\in SO(2)U_{m_s}$. 
Due to weak lower semicontinuity, the last two estimates and the boundedness of $\nabla u_{n}$, we infer
\begin{align}
\int\limits_{\Omega(x_{s},x_{s+1})}|\nabla  u - V_{s}|^4 dz&\leq 
\liminf\limits_{n\rightarrow \infty} \int\limits_{\Omega(x_{s}^{n},x_{s+1}^{n})\cup\Omega(x_{s},x_{s+1})}|\nabla u_n - V^{n}_s|^4 dz\nonumber\\
&\leq \lim\limits_{l\rightarrow \infty}\int\limits_{\Omega(x_{s}^{n_l},x_{s+1}^{n_l})}|\nabla u_{n_l}-V^{n_l}_{s}|^4 dz+\nonumber\\
& \;\;\;\; + \lim\limits_{l\rightarrow \infty}\int\limits_{\Omega(x_{s},x_{s+1})\D \Omega(x_{s}^{n_l},x_{s+1}^{n_l})}|\nabla u_{n_l} - V_{s}^{n_l}|^4 \,dz\nonumber\\
& \leq c\liminf\limits_{l\rightarrow \infty} \int\limits_{\Omega(x_{s}^{n_l},x_{s+1}^{n_l})}\dist(\nabla u_{n_l},SO(2)U_{0})^4 \,dz\nonumber\\
&\;\;\;\;+ c\lim\limits_{l\rightarrow \infty} \mathcal{L}^{1}(\Omega(x_{s},x_{s+1})\D \Omega(x_{s}^{n_l},x_{s+1}^{n_l}))\nonumber\\
& = 0,
\lb{Ms1}
\end{align}
where $\{n_l\}\in\N$ denotes the subsequence which realizes the $\liminf$ in the second inequality.
As a consequence we deduce 
\begin{align*}
\nabla u = V_{s} = Q_{j}U_{m_s} \mbox{ in } \Omega(x_{s},x_{s+1}),
\end{align*}
for some fixed rotation $Q_{j}\in SO(2)$. In a similar way one can conclude that $u_n \go u \mbox{ in } W^{1,4}(Q(x_{s},x_{s+1}))$.
Note that similar estimates hold in the general case for  Hamiltonians satisfying (H1)--(H4) with a $p$-growth assumption.\\

\textit{Step 3b: Degenerate intervals.}
For the case of degenerate intervals we argue similarly. Assume that $x_{s}=x_{s+1}$ for some  $s\in \{1,...,K\}$ but $x_{s}^{n}\neq x_{s+1}^{n}$. 
Without loss of generality, let $x_{s-1}\neq x_{s} \neq x_{s+2}$.
Then, one obtains
\begin{align*}
\int\limits_{\Omega(x_{s-1},x_{s+1})}\dist(\nabla u_n, SO(2)U_{0})^4 dz \lesssim n^{-\alpha} + \frac{n^{\alpha}}{n}+\mathcal{L}^{1}((x_{s}^n,x_{s+1}^n)).
\end{align*}
Here the first energy contribution, $n^{- \alpha}$, originates from the ``good'' stripes $(y_{r,l}^{n},y_{r,l+1}^{n})$ 
with  $r\in \{s-1,s+1\}$ and $l\in \{1,..., K_r^n\}$, respectively. The second contribution comes from the jumps between the stripes
 $(y_{s,l}^{n},y_{s,l+1}^{n})$ and the third one is a simple consequence of
estimate \rf{eq:comp11} reduced to the interval $(x_{s}^{n},x_{s+1}^{n})$. 
We invoke the rigidity estimate on the interval $(x_{s-1}^{n},x_{s+1}^{n})$: There exists $V^{n}_s\in SO(2)U_{m_s}$ such that
\begin{align*}
\int\limits_{\Omega(x_{s-1}^{n},x_{s+1}^n)}|\nabla u_n - V^{n}_s|^4 dz & \leq 
c\int\limits_{\Omega(x_{s-1}^{n},x_{s+1}^{n})}\dist(\nabla u_n,SO(2)U_{0})^4\,dz. 
\end{align*} 
Since the interval $(x_{s-1},x_{s+1})$ has finite, non-degenerate length, the rigidity estimate can be applied with a fixed constant $c$.
As in the first case of non-degenerate intervals, we can find a subsequence such that $V^{n}_s \rightarrow V_s$ and $V_{s}\in SO(2)U_{m_s}$.
Using weak lower semicontinuity, the last estimate and proceeding as in \rf{Ms1} before, one obtains
\bes
\int\limits_{\Omega(x_{s-1},x_{s+1})}|\nabla u -V_s|^4 dz \leq c\liminf\limits_{n\rightarrow \infty}\int\limits_{\Omega(x_{s-1},x_{s+1})}
\dist(\nabla u_n,SO(2)U_{m_s})^4\,dz\rightarrow 0.
\ees
Hence,
\begin{align*}
\nabla u=V_s\mbox{ in } \Omega(x_{s-1},x_{s+1}).
\end{align*}
Therefore, collapsing intervals can be considered irrelevant for the gradient distribution of the limiting function $\nabla u$. 
The limit is determined by the non-degenerate intervals only.\\

\textit{Step 4: Conclusion.}
We identify $V_s=U_0$ if $m_s=0$ and $V_s=Q U_1$ if $m_s=1$. 
Indeed, a combination of the growth property (H3) and the $L^{\infty}$ bounds from (\ref{eq:Linf}) imply uniform $L^{p}$ bounds (and strong convergence) for the gradients $\nabla u_n$ as $n \rightarrow \infty$. This in turn leads to a continuous limit function, $u\in C(\bar{\Omega})$, from which one then obtains that $V_s$ has to be 
rank-one connected to $F_\lambda$. Therefore the statement follows.
\end{proof}

\begin{remark}
\begin{itemize}
\item Note, that the proof of Theorem \ref{T1} relies only on the general assumptions $(H1)-(H4)$ and the definition of the constrained set $\mA_{n,\tau}$ in \rf{Ant}. In particular, it does not make essential use of the specific form of the Hamiltonian \rf{HD}.
\item The converging subsequence and the resulting limiting function depend on the choice of the interpolation.
\end{itemize}
\end{remark}

Analogously, we obtain the following compactness statement:

\begin{proposition}[Compactness]
\label{prop:comp}
Let $F_{\lambda}\in \mathbb{R}^{2\times 2}$, $\lambda \in [0,1]$, be as above. Let $\{u_n\}_{n\in\N} \in \mA_{n,\tau}^{F_{\lambda}}$ be a sequence such that
\begin{align}
\label{eq:en}
 \limsup\limits_{n\rightarrow \infty} H_n^1(u_n)<\infty.
\end{align}
Then there exists a number $K\in\N$ and a (not relabeled) subsequence such that
\begin{itemize}
\item[(i)] $u_n\rightarrow u\ \text{in}\ W^{1,4}(\Omega,\mathbb{R}^2)$,
\item[(ii)] for each $s\in \{1,...,K-1\}$ there exist $m_s\in\{0,1\}$, $x_s\in [-1,1]$  such that
\begin{equation}
\triangledown u(z) \equiv Q^{m_{j}}U_{m_s}
\end{equation}
$\text{ for } z\in \Omega(x_s,x_{s+1}), $
where $Q^0:=\mathrm{Id}$, $Q^1:=Q$ and $x_K=1$,
\item[(iii)] 
$\bigcup\limits_{s=1 }^{K-1}[x_{s},x_{s+1}] = [-1,1].$
\end{itemize}
\end{proposition}

\begin{proof}
This is analogous to the characterization of minimizers: 
Condition (\ref{eq:en}) replaces the direct construction of a comparison function and yields control of the gradient which allows to argue along the same lines as in Theorem \ref{T1}.
\end{proof}

\section{First Order $\Gamma$-Limit and the Limiting Form of the Surface Energy}

\subsection{Setup and Statement of the Result}
In the sequel, we concentrate on an important consequence of Theorem \ref{T1}: The compactness results allow us to determine the limiting form of the surface energy of the sequences satisfying \rf{eq:en} using their previously derived piecewise rigid structure.
We define boundary and internal layer energies, $B^\pm(V_1,V_2,r^*)$,
$C(V_{1},V_{2},r^*)$,  adapted to our situation of specific rank-one
connected matrices via appropriate minimization problems. Let 
\begin{align*}
\mA_n^r:= &&\Big\{u: \Omega_n^r \go\mR^2\Big|\ \mathrm{det}(u(x_2)-u(x_1),u(x_3)-u(x_1))\ge 0\ \text{for all}\nonumber\\
&&\ \{x_1,x_2,x_3\}\subset \Omega_n^r \quad\text{such that}\ \mathrm{diam}(x_1,x_2,x_3)=\sqrt{2}\ \text{and}\nonumber\\
&&\ \mathrm{det}(x_2-x_1,x_3-x_1)\ge 0\Big\}
\end{align*}
and
\begin{equation}
\label{eq:rad}
\begin{split}
\mA_{n,\tau}^r:=& \ \Big\{u\in\mA_n^r\Big|\ u^{i+1j}-u^{ij+1}=-\tau^{i+j+1}\ \text{for all}\ (i,j)\in \Omega_n^r\Big\},\\
&\quad\text{where}\quad\tau^i\in SO(2)\tau\quad\text{for all}\quad i\in [-n,\,n],
\end{split}
\end{equation}
be rescaled versions of the sets \rf{An} and \rf{eq:BData}.
\begin{defi}
\label{defi:bl}
Let $V_1=F_{\lambda}$, for some $0\le\lambda\le 1$ and $V_2$, $V_3$ be
either $U_1$ or $QU_2$ and $\tau:=\begin{pmatrix}-a \\ b \end{pmatrix}$. For functions in the class $\mA_{n,\tau}^r$
we define
\begin{equation}
\label{BL}
\begin{split}
B^+(V_1, V_2, r^*):=& \ \liminf\limits_{n\rightarrow \infty}\min\limits_{\tau_i,u^i}
\Big\{\sum\limits_{i\geq 0}\frac{1}{n}\sum\limits_{j=-n}^{n}h\left(u^{i}_n - u^{i\pm 1}_n,\tau^i_n,\,\tau_n^{i\pm 1}, j\right) : \\ 
&u  \in \mathcal{A}_{n,\tau}^r,\ 
u^{-j j}=V_{1}\begin{pmatrix} -j \\ j \end{pmatrix},  \\
&  u^{i-j j} =V_2\begin{pmatrix} i-j \\ j \end{pmatrix} +r^*, \ i\geq n,\,|j|\le n\Big\}, \\ 
 B^-(V_1, V_2, r^*):=& \ \liminf\limits_{n \rightarrow \infty }\min\limits_{\tau_i,u^i}
\Big\{\sum\limits_{i \leq 0}\frac{1}{n}\sum\limits_{j=-n}^{n}h\left(u^{i}_n - u^{i\pm 1}_n,\tau^i_n,\,\tau_n^{i\pm 1}, j\right):  \\ 
&u  \in \mathcal{A}_{n,\tau}^r,\ 
u^{-j j}=V_{1}\begin{pmatrix} -j \\ j \end{pmatrix} , \\ 
&u^{i-j j} =V_2\begin{pmatrix} i-j \\ j \end{pmatrix} +r^*, \ i\leq -n,\,|j|\le n \Big\}.
\end{split}
\end{equation}
\begin{equation}
\label{IL}
\begin{split}
\displaystyle C(V_2, V_3, r^*):=& \ \liminf\limits_{n\rightarrow \infty}\min\limits_{\tau_i,u^i}
\Big\{\sum\limits_{i\in \Z}\frac{1}{n}\sum\limits_{j=-n}^{n}h\left(u^{i}_n - u^{i\pm 1}_n,\tau^i_n,\,\tau_n^{i\pm 1}, j\right):  \\ 
&u  \in \mathcal{A}_{n,\tau}^r,\ 
u^{i-j j}=V_2\begin{pmatrix} i-j \\ j \end{pmatrix}+r_1 , \ i\leq -n,\,|j|\le n,\\ 
&u^{i-j j}=V_3\begin{pmatrix} i-j \\ j \end{pmatrix} +r_2, \ r^{*}=r_2-r_1, \ i\geq n,\,|j|\le n\Big\}.
\end{split}
\end{equation}
\end{defi}

We remark that as in \cite{BC07} our definitions of the boundary and surface energy layers correspond to minimization problems for which the boundary conditions are ``moved to infinity''. However, in contrast to \cite{BC07} we cannot eliminate the dependence on the parameter $n$ in the densities since we are dealing with two-dimensional energies.\\

The definition of the boundary and internal energy layers allow to introduce a further quantity:
\begin{defi}
Let  $V_{0}=V_K=F_{\lambda}$ for some $\lambda \in [0,1]$ and let
$V_{1},...,V_{K-1}$ belong to the set $\{U_0,\,QU_1\}$. Then we define
\begin{equation}
\displaystyle E^{K}(V_{0},...,V_K):=\inf_{r}\Big\{B^+(V_{0},V_1,r_{0})+\sum_{s=1}^{K-2}C(V_{s},V_{s+1}, r_{s})+B^-(V_{K-1},V_{K},r_{K-1})\Big\},
\label{DE}
\end{equation}
where the infimum is taken over all possible off-set vectors $r=[r_0,...,r_{K-1}]$.
\end{defi}
The main theorem of this section rigorously shows the following asymptotic decomposition of the energy
of any sequence, $\{u_n\}_{n\in \N}$, satisfying the assumptions of Proposition \ref{prop:comp}: 
\begin{equation}
\displaystyle
H_n(u_n)=\la E^{K}(F_\lambda,V_1,...,V_{K-1},F_\lambda)+o(\la )\ \text{as}\ n\rightarrow\infty,
\label{as}
\end{equation}
for some $K\in\N$ and $V_s\in\{U_0,QU_1\},\,s\in \{ 1,...,K-1\}$. This implies that for perturbations of laminar configurations the leading order energy scales as $O(\la)$. Hence, we may interpret the quantity (\ref{DE}) as a surface energy contribution. 
Correspondingly, if $u_n$ is a minimizing sequence to (\ref{DC}), (\ref{eq:newBC}) then
\be
\displaystyle H_n(u_n)=
\la \min_K\min\left\{E^K(F_\lambda,U_0,QU_1,...,F_\lambda),\,E^K(F_\lambda,QU_1,U_0,...,F_\lambda)\right\}+o(\la )
\lb{mas}
\ee
as $n\rightarrow\infty$.
Note that in (\ref{mas}) one needs to minimize an overall number, $K$, of boundary and
internal layers as the exact number is -- a priori -- not given explicitly by Theorem \ref{T1}. 
As a consequence of Theorem \ref{thm:2} we will obtain the expected result $K=3$:
Under our boundary conditions \rf{eq:newBC}, there are, in general, boundary
energy contributions as well as a single interior interface. For the more general
sequences from Proposition \ref{prop:comp}, i.e. for sequences with a finite
surface energy which need not necessarily be minimizers, a finite but
arbitrary number of interior interfaces is possible.\\

With this preparation our main result in this section can be formulated as:

\begin{theorem}[The limiting surface energy]
\label{thm:2}
Let $F_{\lambda}$, $\lambda\in [0,1]$, be as above.
Let $H_n^1(\cdot): \mA_{n,\tau}^{F_{\lambda}} \rightarrow [0,\infty]$ be defined as in \rf{SE}. Then one has 
\begin{align*}
H_{n}^1\stackrel{\Gamma}{\rightarrow} E_{surf} \mbox{ with respect to the } L^{\infty} \mbox{ topology }.
\end{align*}
Here, we have
\begin{align*}
E_{surf}(u):=\left\{ \begin{array}{ll}
E^{K}(F_{\lambda}, \nabla u(x_{1}-,0),..., \nabla u(x_{(K-1)}-,0),F_{\lambda}), \\ 
 \ \ \ \ \ \ \ \mbox{if } u\in W^{1,\infty}_0(\Omega)+F_{\lambda}x, \ \nabla
u\in\{U_0,QU_1\} \text{ in } \Omega(x_j,x_{j+1}),\\
 \ \ \ \ \ \ \  s\in \{1,...,K-1\}; u \mbox{ satisfies the boundary conditions} \\ 
 \ \ \ \ \ \ \  \mbox{prescribed by } (\ref{eq:newBC}),\\
\infty, \ \ \mbox{else},
\end{array} \right.
\end{align*}
where for a right-continuous function $f$ we use the notation $f(x{-} ):= \lim\limits_{x_{i}\downarrow x}f(x_{i})$ and
$W^{1,\infty}_0(\Omega)+F_{\lambda}x:= \left. \{u\in W^{1,\infty}(\Omega) \right|  u=F_{\lambda}x \mbox{ on } \partial_x\Omega \}$.
\end{theorem}
The proof of the $\Gamma$-convergence for the surface energy follows along the lines of the strategy 
introduced by Braides and Cicalese in \cite{BC07}. However, adaptations are needed for our special two-dimensional chain setting: The ``$(1+\epsilon)$''- dimensionality of it causes additional technical difficulties, both in the construction of the recovery sequence for the $\Gamma$-$\limsup$ inequality and in the proof of the $\Gamma$-$\liminf$ inequality. Thus, we start by proving an important auxiliary result in the following subsection. It will play a major role in adapting the strategy of Braides and Cicalese \cite{BC07} to our setting.\\

\subsection{An Auxiliary Result}

In the sequel we introduce one of the crucial techniques used in proving the $\Gamma$-$\liminf$ and 
$\Gamma$-$\limsup$ inequalities. This first technical tool consists of an observation based on averaging and allows to pass from 
a large number of horizontal atomic layers to a smaller number of these without changing the energy much. \\
In order to give the precise statement, we introduce the following quantity. It can be regarded as an intermediate auxiliary functional between $H_n(u_n)$ and $C(V_1,V_2,r_1)$:

\begin{defi}
Let $m, n\in \N$, $n\geq m$. Then for $u_{n}\in \mathcal{A}_{n,\tau}^r$ we set 
\begin{align*}
H^1_{n,m}(u_{n}):=\sum\limits_{i=-n}^{n}\frac{1}{m}\sum\limits_{j=-m}^{m}h\left(u^{i}_{n} - u^{i\pm 1}_{n},\tau^i_{n},\,\tau_{n}^{i\pm 1}, j\right).
\end{align*}
\end{defi}

With the aid of this ``intermediate'' energy we can prove the following central  ``averaging lemma''.

\begin{lem}[Averaging]
\label{lem:av}
Let $m\in \N$ and $\epsilon >0$ be arbitrary but fixed. Let
$u^{ij}_n: \Omega_n^r \rightarrow \mathbb{R}^2$ with
$H_{n,n}(u^{ij}_n)\leq C<\infty$ for all $n\in\N$. Then for any $n\in \N$ with
$n>m\left( 1+ \frac{C}{\epsilon} \right)$ there exists $$u_{\epsilon}: \Omega_{n,m}:= \left\{z\big| \ z= s\begin{pmatrix} 1 \\ 0 \end{pmatrix} + t \frac{1}{\sqrt{2}} \begin{pmatrix} -1 \\1  \end{pmatrix}, \ s\in [-n,n], \ t\in [-m,m] \right\} \rightarrow \mathbb{R}^2$$ 
such that the following statements hold:
\begin{itemize}
\item[(i)] there exists a translation $j_0\in \N$ such that $u_{\epsilon}^{i-j_0 j-j_0}= u^{ij}_n \mbox{ in } \Omega_{n,m},$
\item[(ii)] $H_{n,m}^1(u^{ij}_{\epsilon}) \leq H_{n,n}^1(u^{ij}_n) + \epsilon.$
\end{itemize}
\end{lem}

Here the decisive estimate is given by the last point which allows to pass
from averaging over $n$ points in the vertical direction to averaging over
only $m$ points while creating at most an energy surplus of the size $\eps$.

\begin{proof} 
Fixing $n>m\left( 1+ \frac{C}{\epsilon} \right)$, we construct $u_{\epsilon}$ from $u_{n}$. For this, we subdivide the
parallelogram $\Omega_n^r$ into
$\floor*{\frac{n}{m}}$ disjoint parallelograms (possibly leaving a
  remainder of height less than $m$) of the form $[-2km, 2n-2k m] \times I_k$, $k\in \{0,...,\floor*{\frac{n}{m}}\}$, where $I_k$ is the interval $I_k=[-n+2km,-n+2(k+1)m]$ of the length $|I_k|=2m$. We claim that 
\begin{equation}
\begin{split}
\label{eq:av}
H^{1}_{n,m}\left(u_{n}^{ij}\left(\cdot+\tau^k\right)\right) &\leq H^{1}_{n,n}(u_{n}^{ij}) + \epsilon \mbox{  for some } k\in \left\{0,..., \floor*{\frac{n}{m}}\right\},
\end{split}
\end{equation}
where $\tau^k:=(n-(k+1)m)\begin{pmatrix} -1 \\ 1 \end{pmatrix}$ and $u_{n}^{ij}\left(\cdot+\tau^k\right)$
denotes a vertically translated version of $u^{ij}_{n}(\cdot)$ by $(-n+(k+1)m)$ lattice layers.
In order to observe (\ref{eq:av}), we argue via contradiction. Assuming that the statement of the lemma were wrong, we would obtain
\begin{align*}
H^{1}_{n,m}\left( u_{n}^{ij}\left(\cdot+\tau^k\right)\right)
>H^{1}_{n,n}(u_{n}^{ij}) + \epsilon  \mbox{ for all } k\in \left\{0,..., \floor*{\frac{n}{m}}\right\}.
\end{align*}
Since the strips $I_{k}$, $k\in \{0,..., \floor*{\frac{n}{m}}\}$, are disjoint, this leads to the following estimate:
\begin{equation}
\label{eq:esten}
\begin{split}
nH^{1}_{n,n}(u_{n}^{ij})&\geq
\sum\limits_{k=0}^{\floor{\frac{n}{m}}}m H^{1}_{n,m}
\left(u_{n}^{ij}\left(\cdot+\tau^k\right)\right)
>\left(\frac{n}{m}-1\right)m \left( H^{1}_{n,n}(u_{n}^{ij}) + \epsilon \right)\\
&\geq nH^{1}_{n, n}(u_{n}^{ij}) +(n-m)\epsilon - m H^{1}_{n,n}(u_{n}^{ij}).
\end{split}
\end{equation}
As, however, by assumption $ H^{1}_{n,n}(u_{n}^{ij}) \leq C$, this yields a contradiction:
Since $n>m\left( 1+ \frac{C}{\epsilon} \right)$
and since  $H^1_{n,n}(u^{ij}_n)\leq C$, we infer $m H^{1}_{n, n}(u_{n}^{ij})  \leq  (n- m)
\epsilon$. Hence, (\ref{eq:esten}) cannot be true, which proves (\ref{eq:av}). Defining
$$u_{\epsilon}^{ij}:=u_{n}^{ij}\left(\cdot+\tau^k\right)$$ 
with the corresponding $k\in \left\{0,..., \floor*{\frac{n}{m}}\right\}$, implies the claims of the lemma.
\end{proof}

\subsection{Proof of Theorem \ref{thm:2}}

With the preparation from the previous section, we address the proof of the
$\Gamma$-$\liminf$ inequality: Apart from the ``averaging procedure''
introduced in the previous section a second crucial ingredient of its proof consists 
of a \emph{horizontal ``cutting procedure''}, which allows to modify a given 
configuration with locally small energy by horizontally extending the configuration via 
an appropriate element of $SO(2)U_0\cap SO(2)U_1$ after a certain point. In this sense, the ``cutting
procedure'' in \emph{the horizontal direction} complements the averaging procedure
from the previous section since the latter can be interpreted 
as a ``cutting'' mechanism in \emph{the vertical direction}. Both tools -- the ``cutting'' and ``averaging'' procedures
 -- also play a central role in the construction of the recovery sequence for the $\Gamma$-$\limsup$ inequality later on. 

\begin{proof}[Proof of the $\Gamma$-$\liminf$ inequality]
Let $\{u_{n}\}_{n\in \N}\subset L^{\infty}$ be a sequence such that $u_{n}\rightarrow u$ with respect to the $L^{\infty}$ topology. Without loss of generality, we may assume $H_n^1(u_{n})<\infty.$
According to the compactness result of Proposition \ref{prop:comp}, $u_{n} \rightarrow u$ in $W^{1,4}(\Omega)$ 
and there exists a sequence of points $x_{n}^{1}<...<x_{n}^{K}, \ \{x_{n}^{s}\}_{n\in \N} \subset[-1,1]$ for all $s\in \{1,...,K\}$, as well as limiting points $-1=x^{1}\leq ...\leq x^{K}=1, \ x^{s}\in  (-1,1)$ for all $s\in \{1,...,K\}$, such that (possibly passing to subsequences)
\begin{align}
&h\left( \frac{u^{i}_n - u^{i\pm 1}_n}{\la},\tau^i_n,\,\tau_n^{i\pm 1}, j_l^n\right) \leq \tilde{c} \mbox{ for } i\neq \floor{n x^{s}_{n}}
\mbox{ and } l\in\{-1,0,1\}, \nonumber\\
&x_{n}^{s} \rightarrow x^{s} \mbox{ as } n \rightarrow \infty, \ s\in\{1,...,K\}, \nonumber\\ 
&\nabla u_n \go V_s \in \{U_{0},Q U_{1}\} \mbox{ in } L^4(\Omega(x_{s},x_{s+1})),
\lb{CS}
\end{align} 
where the vertical indexes $j_0^n,j_{\pm 1}^n$ are defined as in \rf{jD}--\rf{BP2}. 
In particular, the limiting points $x^s$, $s\in \{1,...,K\}$, are the only possible jump points of the gradient of $u$.
Since we may always pass to the infimizing sequence of the energy, we may further assume w.l.o.g. that the
statements of \rf{CS} hold for the whole sequence.
We remark that the $x^{s}_{n}$ are possibly degenerate in the sense that $x^{s}=x^{s+1}$ for some $s\in \{1,...,K\}$. However, in the sequel we only consider the case of non-degenerate points $x^s$, and briefly comment on the necessary modifications in the case of degenerate points at the end of the proof.\\
In order to pass from the coordinates in $\Omega$ to the integer coordinates in $\Omega_n^r$, we keep track of the number of atoms between the respective jumps of the gradient in the $n$-th iteration step by defining sequences $\{h_{n}^{s}\}_{n\in \N}$, $h_{n}^{s} \in\N$, and $\{k_{n}^{s}\}_{n\in\N}$, $k_{n}^{s} \in\N$, by setting
\begin{align}
\label{eq:count}
\displaystyle \lim_{n\go\infty}\la h_{n}^{s} = \frac{x_{s+1}-x_{s}}{2} \ \mbox{ and } \
 k_{n}^{s} =-n+2\sum\limits_{i=1}^{s-1}h_{n}^{i},
\end{align}
for $s\in\{1,...,K-1\}$. 
This allows to rewrite the energy as
\begin{align*}
 H_{n}^1(u_{n}) = \la\sum\limits_{i,j=-n}^{n}h \left( \frac{u^{i}_n - u^{i\pm 1}_n}{\la},\tau^i_n,\,\tau_n^{i\pm 1}, j\right)
 = H_{n}^{1} + \sum\limits_{s=2}^{K-1}H_{n}^{s} + H_{n}^{K},
\end{align*}
with
\beas
 H_{n}^{1} &:=&  \la \sum\limits_{i=-n}^{h_{n}^{1}}\sum\limits_{j=-n}^{n} h\left(\frac{u^{i}_n - u^{i\pm 1}_n}{\la},\tau^i_n,\,\tau_n^{i\pm 1}, j\right), \\
H_{n}^{s} &:=&  \la \sum\limits_{i=k_{n}^{s-1}-h_{n}^{s-1}}^{k_{n}^{s-1}+h^{s}_{n}}\sum\limits_{j=-n}^{n}
h\left( \frac{u^{i}_n - u^{i\pm 1}_n}{\la},\tau^i_n,\,\tau_n^{i\pm 1},j\right),\quad s=2,...,K-1,\\
 H_{n}^{K} &:=&  \la \sum\limits_{i=k_{n}^{K-1}-h_{n}^{K-1}}^{k_{n}^{K-1}}\sum\limits_{j=-n}^{n}h\left( \frac{u^{i}_n - u^{i\pm 1}_n}{\la},\tau^i_n,\,\tau_n^{i\pm 1}, j\right).
\eeas
At this point, we would like to relate our energy contributions $H_{n}^{1},...,H_{n}^{K}$ 
to the interfacial and boundary layer energies from Definition \ref{defi:bl}. As (rescaled versions of) our functions $u_{n}$ are very close to being admissible for the definition of these energies, we only have to modify them slightly: A first ansatz would be to use $u_n$ for the region close to the expected jump and then extend this deformation by the correct element $V_s$ of one of the energy wells $SO(2)U_0$ or  $SO(2)U_1$ (c.f. Braides and Cicalese, \cite{BC07}).
Since, however, our gradient sequence only converges in $L^p$ and \textit{not uniformly}, this extension might cause an error of $O(1)$. \\
In order to avoid this difficulty, we use the (uniform in $n$) bound $H_{n}^1(u_n)\leq c<\infty$ which allows to argue similarly as in Step 2 of Theorem \ref{T1}. The result of the next lemma implies the $\Gamma$-$\liminf$ inequality.

\begin{lem}
\label{lem:cut}
For any $s\in \{1,...,K\}$ it holds 
\begin{align}
H^{s}_n \geq \left\{ \begin{array}{ll} B^{+}(F_{\lambda},V_1,r_0) -w(n) &\mbox{ if } s=1, \\
C(V_s, V_{s+1}, r_s) -w(n) &\mbox{ if } s\in \{2,...,K-1\},\\
B^{-}(V_{K-1},V_K,r_{K-1}) -w(n) & \mbox{ if } s=K,
 \end{array} \right.
\lb{EEs}
\end{align} 
where $w(n) \rightarrow 0$ as $n \rightarrow \infty$.
\end{lem}

We present the strategy of the proof in greatest detail for the left boundary layer.
The arguments for the interior interfaces and the right boundary layer follow along the same lines.
We indicate the main differences and how to overcome additional difficulties.

\begin{proof}
\emph{Step 1: The left boundary layer.}
By the estimates on the simultaneously good points from the proof of Theorem
\ref{T1}, we infer that for any $0<\alpha<1$ there exists an integer $r^{\alpha,1}_{n}\in ( -n, -n+ n^{\alpha}]$
such that in $\Omega(\la r^{\alpha,1}_{n},\la r^{\alpha,1}_n+\la)$ the function $u_n$ satisfies estimates of the form \rf{Ms3}, i.e. there exists a rotation $Q_n$ such that
\begin{equation}
\label{eq:alpha}
\begin{split}
&||\nabla u_n - Q_nV_1||_{C(\Omega(\la r^{\alpha,1}_{n},\la r^{\alpha,1}_n+\la))} \lesssim n^{- \frac{\alpha}{4}},
\\
& |n\delta^{r^{\alpha,1}_{n}\pm 1}_n|\lesssim n^{- \frac{\alpha}{4}}.
\end{split}
\end{equation}
In particular, the first equation yields an estimate on the vertical extension vectors:
\begin{align*}
|Q_n\tau- \tau^{r_n^{\alpha,1}}_n|\leq n^{- \frac{\alpha}{4}}.
\end{align*}
We would like to exploit this, in order to obtain a test function for the minimization problem (\ref{BL}). For this we use a strategy based on defining a test function as (a rescaled version of) $u_n^{ij}$ for the first $r^{\alpha,1}_n$ horizontal steps and then extending it by an appropriately translated version of the affine function $V_1 x$. However, we have to be careful not to violate the non-interpenetration condition (\ref{An}). In the sequel, we provide the details of the construction. \\
We start by remarking that equation \rf{eq:Linf} yields
\begin{align}
|n\delta^{i\pm 1}_n|\leq c \mbox{ for all } i\in [0, r^{\alpha,1}_n].
\lb{des}
\end{align}
Making use of the left boundary data, i.e. $u_n^{i,0}=F_{\lambda}\begin{pmatrix} i \\ 0 \end{pmatrix}$ for $i\leq -n$, it is then possible to deduce good closeness properties between $\tau$ and $\tau^{r^{\alpha,1}_n}_n$ via a telescope sum:
\begin{align}
\label{eq:taudiff}
|\tau- \tau^{r^{\alpha,1}_n}_n| \leq  \sum\limits_{i=-n}^{r^{\alpha,1}_n}|\tau^{i}_n - \tau^{i+1}_n|  = \sum\limits_{i=-n}^{r^{\alpha,1}_n}|\delta_n^{i+1}| \lesssim n^{\alpha -1} \lesssim n^{- \alpha-\delta},
\end{align}
for $0<\alpha<\frac{1}{2}-\delta$ and $0<\delta\ll1$.
We further claim, that 
\begin{align*}
|V_{1}-Q_nV_{1}|\lesssim  n^{-\frac{\alpha}{4}}.
\end{align*}
Indeed, from the previous estimates we obtain 
\begin{align*}
\left|V_{1}\begin{pmatrix} 1\\ -1 \end{pmatrix}-Q_n V_1 \begin{pmatrix} 1\\-1 \end{pmatrix} \right| &\leq \left| V_1 \begin{pmatrix} 1\\-1 \end{pmatrix} - \tau^{r^{\alpha,1}_n}_n\right| + \left|Q_n V_1 \begin{pmatrix} 1\\-1 \end{pmatrix} -\tau^{r^{\alpha,1}_n}_n \right| \\ &=
|\tau-\tau^{r^{\alpha,1}_n}_n| + |Q_n \tau  -\tau^{r^{\alpha,1}_n}_n |\\
&\lesssim   n^{-\frac{\alpha}{4}}.
\end{align*}
Combining this with the fact that 
\begin{align*}
\left|V_{1}\begin{pmatrix} 1\\ 1 \end{pmatrix}-Q_n V_1 \begin{pmatrix} 1\\1 \end{pmatrix} \right| 
&= \left|\begin{pmatrix} a\\ b \end{pmatrix}-Q_n \begin{pmatrix} a\\b \end{pmatrix} \right| \\
&= \left|\tilde{Q}\begin{pmatrix} a\\ b \end{pmatrix}- \tilde{Q}Q_n  \begin{pmatrix} a\\b \end{pmatrix} \right|
=\left|\begin{pmatrix} a\\ -b \end{pmatrix}-Q_n \begin{pmatrix} a\\-b \end{pmatrix} \right|  \\
&= \left|V_{1}\begin{pmatrix} 1\\ -1 \end{pmatrix}-Q_n V_1 \begin{pmatrix} 1\\-1 \end{pmatrix} \right|
\lesssim  n^{-\frac{\alpha}{4}},
\end{align*}
yields the claim. Here $\tilde{Q}$ denotes the rotation matrix mapping the vector $\begin{pmatrix} a\\ b \end{pmatrix}$ to the vector  $\begin{pmatrix} a\\ -b \end{pmatrix}$. In the second line we made use of the commutativity of $SO(2)$.\\
Thus, this leads to
\begin{equation}
\label{eq:graddiff}
\tag{\ref{eq:taudiff}'}
\begin{split}
|\nabla u^{r^{\alpha,1}_n0}_n-V_{1}\pm n^{\alpha}(\tau^{r^{\alpha,1}_n}_n-\tau)| &\leq 
 n^{-\frac{\alpha}{4}}+ |n^{\alpha}(\tau^{r^{\alpha,1}_n}_n-\tau)|\\ &\lesssim  n^{-\frac{\alpha}{4}} + n^{\alpha}n^{- \alpha - \delta}\\
& \lesssim  \max\{n^{-\frac{\alpha}{4}},n^{-\delta}\}.
\end{split}
\end{equation}
Using Lemma 4.1 for $m(n):=n^{\alpha}$ and (after a possible vertical translation of the original function, which we suppress in the sequel) 
for $- m\le j\le m$, we set 
\begin{equation}
\label{eq:testf}
\begin{split}
\tilde{u}^{i-j j}_{m} := \left\{
\begin{array}{ll}
\frac{u^{i-jj}_{n}}{\la} &\mbox{ for } 0\leq i<r^{\alpha,1}_n,\\
 V_1\begin{pmatrix} i-j\\ j \end{pmatrix} - V_1\begin{pmatrix} r^{\alpha,1}_{n}\\ 0  \end{pmatrix}
+ \frac{u_{n}^{{r^{\alpha,1}_n0}}}{\la} &\mbox{ for } i \geq r^{\alpha,1}_n.\\
\end{array}
\right.
\end{split}
\end{equation}
Therefore, invoking Lemma \ref{lem:av} with e.g. $\epsilon(n):=n^{-\delta}$, it holds
\begin{align}
\label{eq:enest}
H^1_n \geq H^1_{n,m}(\tilde{u}_m)-\epsilon(n)=H^1_{m,m}(\tilde{u}_m)- \epsilon(n).
\end{align}
We stress that thus $\tilde{u}_{m}$ is an admissible test function -- satisfying in particular (\ref{eq:rad}) -- for the minimum problem
defining $B^{+}(F_{\lambda},V_1,r_{0}^{m})$ on the scale $m=n^{\alpha}$, where
\begin{align*}
r_{0}^{m(n)}:=\frac{u_{n}^{r^{\alpha,1}_n 0}}{\la} - V_1\begin{pmatrix} r^{\alpha,1}_n\\ 0 \end{pmatrix}.
\end{align*}
Using \rf{eq:graddiff}, we estimate
\begin{align}
H_{n}^{1} & \geq  \lambda_m  \sum\limits_{i=-m}^{m}\sum\limits_{j=-m}^{m}h\left(\frac{u^{i}_n - u^{i\pm 1}_n}{\la},\tau^i_n,\,\tau_n^{i\pm 1}, j\right)- \epsilon(n)
\nonumber\\
&\geq H_{m,m}^1(\tilde{u}_{m})- w_1(n)- \epsilon(n) \nonumber\\
&\geq B^{+}(F_{\lambda}, V_1, r_{0}^{m}) - w_1(n)- \epsilon(n),
\lb{1BL}
\end{align}
\begin{align*}
|w_1(n)|\lesssim \left\|\nabla u_n- V_1\right\|_{L^{\infty}(\Omega(\la r^{\alpha,1}_{n},\la r^{\alpha,1}_n+\la))}+o(1)\go 0
\end{align*}
as $n\go\infty$. Thus, in the limit $n \rightarrow \infty$, estimate \rf{1BL} implies:
\be
\displaystyle \liminf\limits_{n\rightarrow \infty} H_n^1\ge \inf_{r_0} B^+(F_\lambda,V_1,r_0).
\lb{EEs1}
\ee
\emph{Step 2: Internal layers.}
For the remaining intervals we argue analogously: For each $s\in\{2,...,K-1\}$ and each $0<\alpha\leq \frac{1}{2}-\delta$ as above,
there exist integers ${l}^{\alpha,s}_{n}\in [k_{n}^{s}-n^{\alpha}, k_{n}^{s} ]$, and 
 ${r}^{\alpha, s}_{n}\in  [k_{n}^{s}, k_n^s+n^{\alpha} ]$ such that $\nabla u_n^{ij}$ stays $O(n^{-\alpha/4})$ close to certain rotations of $V_s$ and $V_{s+1}$ in the domains 
$\Omega({l}^{\alpha,s}_{n} \la ,{l}^{\alpha,s}_{n}\la+\la)$ and
$\Omega({r}^{\alpha,s}_{n}\la,{r}^{\alpha s}_{n}\la+\la)$.\\
Moreover, by an argument which is similar to the one used for the left boundary layer,
we may without loss of generality assume that the deformation gradients (and
in particular the associated rotations) on the left and on the right hand side
of the boundary layer are $O(n^{-\alpha/4})$  close to $V_{s-1}$ and
$V_{s}$, respectively. Indeed by similar considerations as before, we may
first assume that there exists a single rotation $Q_n^s\in SO(2)$ such that
$\nabla u^{ij}_n$ is $O(n^{-\alpha/4})$ close to $Q_n^s V_{s-1}$ and
$Q_n^s V_{s}$ on the left and right hand side neighborhoods of the jump
interface, respectively. Secondly, by switching from $u_{n}^{ij}$ to
$\bar{u}^{ij}_n:=(Q^{s}_n)^{-1} u^{ij}_n$ in the respective $n^{\alpha}$
neighborhoods of the interface, we may assume that the deformation is close to
$V_{s-1}$ and $V_{s}$ on the left and right hand sides of the jump layer,
respectively. In the sequel, we assume that -- if necessary -- the appropriate rotation has already been carried out and omit the bars in the notation.\\
Again invoking Lemma \ref{lem:av} and setting $m(n):=n^{\alpha}$, 
(after a possible vertical translation) we define a new horizontally truncated deformation by
\begin{align*}
\tilde{u}^{i-j j}_{m}:= \left\{
\begin{array}{ll}
V_{s-1}\begin{pmatrix} i-j\\ j \end{pmatrix} + V_{s-1}\begin{pmatrix}{l}_n^{\alpha,s}\\ 0 \end{pmatrix}+ \frac{u_{n}^{{l}_n^{\alpha,s} 0}}{\la} &
\mbox{ for } i \leq {l}_n^{\alpha,s},\\
\frac{u_{n}^{i-jj}}{\la}  & \text{ for } {l}_n^{\alpha,s}<i<{r}_n^{\alpha,s},\\
V_s\begin{pmatrix} i-j\\ j \end{pmatrix} -V_s\begin{pmatrix} {r}_n^{\alpha,s} \\ 0 \end{pmatrix} + \frac{u_{n}^{{r}_n^{\alpha, s}0}}{\la}
& \mbox{ for } i \geq {r}_n^{\alpha,s},
\end{array} \right.
\end{align*}
for $-m \leq j \leq m$ and note that $\tilde{u}_{m}\in\mathcal{A}_{m,\tau}^r$. Thus, $\tilde{u}_m$ is an admissible test function for the internal energy layer
$C(V_{s-1},V_s, r^{m}_{s-1})$ on the scale $m=n^{\alpha}$. Estimating the energies thus yields 
\begin{align}
H_{n}^{s} \geq & \ \lambda_m \sum\limits_{i={l}_{n}^{\alpha,s}+1}^{{r}^{\alpha,s}_n}\sum\limits_{j=-m}^{m}h\left(\frac{u^{i}_n - u^{i\pm 1}_n}{\la},\tau^i_n,\,\tau_n^{i\pm 1}, j\right) -w_{s}(n)-\epsilon(n) \nonumber\\
=& \ H_{m,m}^{s}(\tilde{u}_{m})- w_s(n)- \epsilon(n)\nonumber\\
\geq& \ C(V_{s-1}, V_s, r_{s-1}^{m}) - w_s(n)- \epsilon(n),
\lb{EEs2} 
\end{align}
where 
\begin{align*}
r_{s-1}^{ m(n)}:=- V_{s-1}\begin{pmatrix} {l}_n^{\alpha,s}+1\\ 0\end{pmatrix}- \frac{u_{n}^{{l}_n^{\alpha,s}+10}}{\la}
-V_s\begin{pmatrix} {r}_n^{\alpha,s} \\ 0\end{pmatrix} + \frac{u_{n}^{{r}_n^{\alpha,s}0}}{\la}.
\end{align*}
and 
\begin{align*}
|w_s(n)|\lesssim & \ \left\| \nabla u_n-V_{s-1}\right\|_{L^{\infty}(\Omega(\lambda_{n} l^{\alpha,s}_{n},\lambda_{n} l^{\alpha,s}_{n}+\la ))}\\
&+ \left\| \nabla u_n-V_s\right\|_{L^{\infty}(\Omega(\lambda_{n}r^{\alpha,s}_{n} ,\lambda_{n}r^{\alpha,s}_{n}+\la))}+o(1)\go 0,
\end{align*}
as $n\go\infty$.\\

\emph{Step 3: The right boundary layer.}
Finally, for the right boundary layer, we argue as in the case of the left
boundary layer. Denoting by ${l}^{\alpha,(K-1)}_{n}$,  the
corresponding integer in the interval $[n-n^{\alpha},n]$, recalling Lemma \ref{lem:av} and (after a possible vertical translation) setting 
\begin{align*}
\tilde{u}^{i-j j}_{m}:= \left\{
\begin{array}{ll}
V_{K-1}\begin{pmatrix} i-j\\ j \end{pmatrix}+ V_{K-1}\begin{pmatrix} l^{\alpha,(K-1)}_{n} \\0\end{pmatrix} + \frac{u_{n}^{{l}^{\alpha,(K-1)}_{n}0}}{\la} 
& \mbox{ for } i \leq {l}^{\alpha,(K-1)}_{n},\\
\frac{u_{n}^{i-jj}}{\la} &  \mbox{ for } i> {l}^{\alpha,(K-1)}_{n},
\end{array} \right.
\end{align*}
for $m=n^{\alpha}$ and $-m\leq j \leq m$, implies 
\begin{align}
H_{n}^{K} &\geq  \lambda_m \sum\limits_{i=n-m}^{n}\sum\limits_{j=-m}^{m}
h\left(\frac{u^{i}_n - u^{i\pm 1}_n}{\la},\tau^i_n,\,\tau_n^{i\pm 1}, j\right)-\epsilon(n)\nonumber\\
& \geq
H_{n,m}(\tilde{u}_{m})-w_K(n)-\epsilon(n)\nonumber\\
&\geq B^{-}(V_{K-1},F_{\lambda},r_{K}^m)-w_K(n)-\epsilon(n),
\lb{EEs3}
\end{align}
where 
\bes
r_{K}^{ m(n)}:=V_{K-1}\begin{pmatrix} {l}^{\alpha,(K-1)}_{n}\\ 0 \end{pmatrix} + \frac{u_{n}^{{l}^{\alpha,(K-1)}_{n} 0}}{\la},
\ees 
and $|w_K(n)|\lesssim \left\|\nabla
u_n-V_{K-1}\right\|_{L^{\infty}(\Omega(\lambda_{n}l_{n}^{\alpha,(K-1)},\lambda_{n}l_{n}^{\alpha,(K-1)}+\la))}+o(1)$
as $n\go\infty$.
Combining the estimates \rf{EEs1}--\rf{EEs3}, we finally infer the desired inequality \rf{EEs}.
Hence the lemma and therefore, the $\Gamma$-$\liminf$ inequality in the case of non-degenerate intervals, are proved.
\end{proof}

In the case of degenerate intervals we mainly argue along the same lines as
above. In this case the sequence $u_n$ may have
more transition layers than the limiting function $u$. As degenerate intervals possibly yield additional transitions 
(with the length scale $\alpha$ possibly chosen accordingly to the
  degeneracy of their lengths as $n\go\infty$ -- but always keeping $\alpha$ bounded from below by an $n$-independent constant) one may rely on the triangle inequality for the boundary layer energies, e.g. in the form of the estimate:
\begin{align*}
\inf\limits_{r_{s-1}} C(V_{s-1},V_s,r_{s-1}) + \inf\limits_{r_s} C(V_s,V_{s+1},r_s) \geq \inf\limits_{r_{s-1}}C(V_{s-1},V_{s+1},r_{s-1}).
\end{align*}
Hence, the $\Gamma$-$\liminf$ inequality also holds in this setting. 
\end{proof}

\begin{remark}
For later reference, we summarize the essential modification steps which were used in the $\Gamma$-$\liminf$ inequality and refer to them as a ``cutting procedure''. As outlined in the proof of the $\Gamma$-$\liminf$ inequality, it involves 
\begin{itemize}
\item finding integer points close (more precisely, $n^{\alpha}$-close) to the expected jump layers at which the configuration transforms from one of the wells to the other well (and possibly back to the original well), c.f. equation (\ref{eq:alpha}),
\item transferring the good estimates on neighboring $\tau^i_n$, c.f. \rf{des}, to the $\alpha$-scale in order to avoid self-interpenetration of the material, c.f. equations (\ref{eq:taudiff}) and (\ref{eq:graddiff}) which is possible for any $\alpha<1/2$,
\item horizontally gluing the right rotation to the prescribed sequence at the respective
  $\alpha$-close points, c.f. (\ref{eq:testf}),
\item using Lemma \ref{lem:av} in order to pass from the original sequence,
  which had a scale $\sim n$, to a modified sequence which is reduced to a
  scale $m=n^{\alpha}$ in both the horizontal and the vertical directions. 
\end{itemize}
We emphasize that both the second and the fourth steps are essential in preserving the rescaled non-interpenetration condition  (\ref{eq:rad}), 
c.f. (\ref{eq:graddiff}).
\end{remark}

Keeping the previous comments in mind, we continue with the proof of the $\Gamma$-$\limsup$ inequality. Again, this is
slightly more involved than the strategy proposed in~\cite{BC07}, as we again have to preserve the
admissibility conditions, i.e. $u_n \in \mA_{n,\tau}$.
In particular, we are confronted with the presence of an $n$-dependence in the functional which we analyze. 
As before, the key tools consist of a version of the averaging lemma (Lemma 4.1) and the ``cutting procedure'' introduced in Lemma 4.2.

\begin{proof}[Proof of the $\Gamma$-$\limsup$ inequality]
\emph{Step 1: Preparation.}
Let $u$ be such that $E_{surf}(u)<\infty$. Then there exist $K\in\N,\,V_s\in\{U_0,\,QU_1\},\,s\in \{ 1,...,K-1 \}$
and $-1=x_1<x_2,...,x_{K-1}<x_K=1$ such that 
\bes
\nabla u(z)=V_s\quad\text{for}\quad z\in Q(x_s,\,x_{s+1}),\ s\in \{1,...,K-1 \}.
\ees
For a fixed but arbitrary $\eps>0$ let $r=[r_0,...,r_K]$ be such that 
\begin{equation}
\label{eq:offset}
B^+(F_{\lambda},V_0,r_{0})+\sum\limits_{s=2}^{K-2}C(V_{s-1},V_s, r_{s})+B^-(V_{K-1},F_{\lambda},r_{K})
\le E_{surf}(u)+\eps.
\end{equation}
By definition of the internal layer energies, for every $s\in
\{2,...,K-2\}$ there exist {\it an infimizing} subsequence $\{n_{m_s}\}_{m_s \in \N} \subset \N$
with $n_{m_s} \rightarrow \infty$ and functions $u_{n_{m_s}}\in \mathcal{A}_{n_{m_s},\tau}^r$ such that
\begin{align}
\lim\limits_{m_s \rightarrow \infty} H_{n_{m_s}}^s\left(\lambda_{n_{m_s}}u_{n_{m_s}}\left( \frac{\cdot }{\lambda_{n_{m_s}}} \right) \right)=C(V_s,V_{s+1},r_s).
\lb{is}
\end{align}
Analogous statements hold for the boundary layers.
We would like to use the subsequences $\{n_{m_s}\}_{m_s\in \N}$ and the
functions $u_{n_{m_s}}$ which realize the $\liminf$ in Definition
\ref{defi:bl} in order to define a recovery sequence for the $\Gamma$-$\limsup$
inequality.  For this, however, we have to fill the gaps in the subsequence $u_{n_{m_s}}$ in a way which does not raise the energies contributing to $\limsup\limits_{n\rightarrow \infty} H_n(u_n)$. In order to deal with this difficulty, we make use of an ``energy partitioning''  or ``averaging argument'' in the spirit of Lemma \ref{lem:av}. \\

\emph{Step 2: Averaging and collecting properties of the layer energies.}
In the sequel, we focus on a single transition layer, say on the transition
layer $C(V_1,V_2,r_1)$. All the other boundary and internal layer
contributions can be treated analogously.  We denote the subsequence realizing the
$\liminf$ in Definition \ref{defi:bl} by $\{n_{m}\}_{m\in\N}$ and the
associated functions by $u_{n_m}$. These functions are defined on $\Z \times
[-n_m,n_m]$ and fulfill the rescaled admissibility condition (\ref{eq:rad}).
 
We centrally use the following variant of Lemma \ref{lem:av}:

\begin{lem}[Energy Partitioning/Averaging]
\label{lem:en}
Let $n\in \N$ and $\epsilon>0$ be arbitrary. Then there exists $\bar{n}\in \{n_m\}_{m\in \N}$ with $\bar{n}\gg n$, $\bar{n}=\bar{n}(\epsilon,n)$, such that there exists a function $u_{\epsilon}:\mathbb{Z} \times [-n,n] \rightarrow \mathbb{R}$,  $u_{\epsilon}\in  \mathcal{A}_{n,\tau}^r$ satisfying 
\begin{enumerate}
\item[(i)]$u_{\epsilon}\left(\cdot-\begin{pmatrix} l\\ 0 \end{pmatrix}\right)\in \mathcal{A}_{n,\tau}^r$, for all $l\in\mathbb{Z}$,
\item[(ii)] $u_{\epsilon}^{i-j j}=V_{1}\begin{pmatrix} i-j \\ j \end{pmatrix} , \ i\leq -\bar{n}$ and  
$u_{\epsilon}^{i-j,j} =V_2\begin{pmatrix} i-j \\ j \end{pmatrix} + r_1, \ i\geq \bar{n}$,
\item[(iii)] $ H^{1}_{\bar{n},n}(u_{\epsilon}^{ij}) \leq H^{1}_{\bar{n},\bar{n}}(u_{\bar{n}}^{ij})+\epsilon. $
\end{enumerate}
\end{lem}

\begin{proof}
The proof of the Lemma essentially follows along the same lines as the proof
of Lemma \ref{lem:av} with $\Omega_n^r$ and $\Omega_{n,m}$ replaced by
$\mathbb{\Z}\times [-\bar{n},\bar{n}]$ and $\Z \times [-n,n]$: For a given $n\in \N$ we argue as in Lemma 4.1 and choose 
$\bar{n}>n\left(1+ 2\frac{C(V_1,V_2,r_1)}{\epsilon} \right)$, $\bar{n}\in\{n_m\}_{m\in\N}$. 
This yields the existence of $u_{\epsilon}$ satisfying condition (iii). Condition (ii) follows from the definition of $u_{\bar{n}}$. 
Finally, condition (i) is a consequence of the fact that integer valued, horizontal translations do not change the admissibility of the sequence. 
\end{proof}

Lemma \ref{lem:en} thus allows us to extend the infimizing subsequences
 $\{u_{n_{m_s}}\}$ in \rf{is} for each $s \in \{1,...,K-1\}$ to full sequences $\{u_n^s\},\ n\in\N$, such that for
 $\eps>0$ fixed  and all sufficiently large $n$ one has
\begin{align}
\label{eq:enn}
H_{n}^1\left( \lambda_n u_n^{s }\left( \frac{\cdot}{\la} \right) \right)\leq\left\{ \begin{array}{ll} 
B^+(F_{\lambda},V_1,r_{0})+ c\epsilon &\mbox{ if } s=0,\\
C(V_s,V_{s+1},r_s) +c\epsilon &\mbox{ if } s \in \{1,...,K-2\}, \\
B^-(V_{K-1},F_{\lambda},r_{K-1}) +c\epsilon &\mbox{ if } s=K-1.
\end{array} \right.
\end{align}
We now prepare for patching together the various $u_n^s$ originating from
the different internal and boundary layer energies. In this context the ``cutting procedure''  summarized in Remark 4.1 plays an essential role.\\

The next lemma shows that  after slight modifications which, however, preserve the
estimates \rf{eq:enn} (up to additional $O(\eps)$ terms), the sequence $\{u_n^s\}$ can be chosen to converge
to a function $u_s$ with $\nabla u_s \in \{V_s, V_{s+1}\}$ whose gradient only has a single jump layer: 

\begin{lem}
\label{lem:enI}
Let $0<\alpha< \frac{1}{2}-\delta$, $0< \delta \ll1 $ and $0<\alpha<\beta < \alpha+\delta/2$ be fixed but arbitrary. 
For each $s \in \{1,...,K-1\}$ it is possible to modify the sequence
$\{u_n^s\}_{n\in\N}$ given above such that for $m(n)\sim n^{\beta}$
the new sequence $\{\tilde{u}_{m}^s\}_{m\in \N},\ \tilde{u}_{m}^s\in\mathcal{A}_{m,\tau}^r$ remains admissible and has exactly one transition
layer between $V_s$ and $V_{s+1}$ of the width $O(n^{\alpha})$ (see Fig. 3) which is (without
loss of generality) centered at 0. More precisely, there exists $i_n^s\in\N$ and
a constant $c$ such that
\bea
&&i_n^s \lesssim n^{\alpha},\nonumber\\
&&\nabla \tilde{u}_{m}^{ s,ij}=V_s\quad\text{for}\quad i<-i_n^s,\quad\nabla\tilde{u}_{m}^{ s,ij}=V_{s+1}\quad\text{for}\quad i>i_n^s
\lb{Ms16}
\eea
and $-m \le j\le m$.
Moreover, $\tilde{u}_{m}^s$ preserves the energy estimates (\ref{eq:enn}) and 
\bea
\lb{Ms14}
&&\lambda_{m} \tilde{u}_{m}^s\left( \frac{\cdot}{\lambda_m} \right) \go \tilde{u}^s\quad\text{in}\quad W^{1,4}(\Omega_s)\cap L^{\infty}(\Omega_s),
\eea
with $\nabla \tilde{u}^s=V_s$ for $x<0$ and $\nabla \tilde{u}^s=V_{s+1}$ for $x>0$.
Here we use the notation $\Omega_s:=\Omega\left(\frac{x_{s-1}-x_s}{2},\frac{x_{s+1}-x_s}{2}\right)$.
\end{lem}

\begin{proof}
By choosing $\bar{n}$ sufficiently large compared to $n$ in the application of Lemma \ref{lem:en}, we may, without loss of generality, assume that any horizontal layer, i.e. any distance between successive ``bad'' points'' $x_n^{j,s}$, $j\in\{1,...,K^{s}\}$, is either of a size $\lesssim n^{\alpha}$ or at least of the size $\sim n^{\beta}$ and that $u_n^s$ is defined on a strip of the size
$\bar{n}\times n$. By an argument similar to the one used in the proof of the $\Gamma$-$\liminf$ inequality, for each element $u_n^s$ of the sequence, one may choose integers $l_n^{\alpha ,s}\in\Z$,
$r_n^{\alpha,s}\in\Z$ which are in an $n^{\alpha}$ neighborhood of an interface between $V_s$ and $V_{s+1}$ and where $\nabla u_n^s$ is
$O(n^{-\alpha/4})$ close to $V_s$ and $V_{s+1}$, respectively. Indeed, in order to obtain this reduction we argue as in Step 2 of the proof of Lemma \ref{lem:cut}: By the considerations carried out in the proof of Lemma \ref{lem:cut} we may firstly assume that the gradient is in a neighborhood of $Q_n V_s$ and $Q_n V_{s+1}$ on the left and right sides of the jump interface for some $Q_n\in SO(2)$; secondly by defining $\tilde{u}^{ij}:= Q_n^{-1} u^{ij}$ we may then assume $Q_n=Id$ without violating the admissibility of the sequence.
At such positions we cut  $u_n^s$ and extend it by a deformation gradient given by
$V_s$ or $V_{s+1}$ as in the proof of the $\Gamma$-$\liminf$-inequality. More precisely, (after a possible vertical translation) we define 
\begin{align*}
\tilde{u}_{m}^{s,i-jj} := \left\{
\begin{array}{ll}
V_{s}\begin{pmatrix} i-j\\ j \end{pmatrix} + V_{s}\begin{pmatrix}{l}_n^{\alpha,s}\\ 0 \end{pmatrix}+ \frac{u_{n}^{{l}_n^{\alpha,s}0}}{\la} &
\mbox{ for } i \leq {l}_n^{\alpha,s}+1,\\
\frac{u_{n}^{i-jj}}{\la} & \mbox{ for } {l}_n^{\alpha,s}<i<{r}_n^{\alpha,s},\\
V_{s+1}\begin{pmatrix} i-j\\ j \end{pmatrix} -V_{s+1}\begin{pmatrix} {r}_n^{\alpha,s} \\ 0 \end{pmatrix} + \frac{u_{n}^{{r}_n^{\alpha, s}0}}{\la} &
\mbox{ for } i \geq {r}_n^{\alpha,s},
\end{array} \right.
\end{align*}
for a certain $m(n)\sim n^{\beta}$ and $-m\leq j \leq m$. Thus, we consider a
transition layer of the size $O(n^{\alpha})$ but instead of cutting out a
square of the size $n^{\alpha}$, we
work on a square of the size $O(n^{\beta})$. This allows to obtain the desired $L^{\infty}$ convergence in (\ref{Ms14}). 
By adapting the exact value of $m$ (i.e. possibly correcting it by a multiplicative constant), it is possible to satisfy \rf{Ms14} on the domain $\Omega_s$.
Under the stated constraints on $\alpha, \beta$ and by invoking Lemma \ref{lem:av}, the
construction remains admissible (in particular (\ref{eq:rad}) can be
satisfied). Moreover, its energy is controlled by \rf{eq:enm} given below, i.e. it satisfies an analog of the energy bound \rf{eq:enn} 
for the modified functions $\tilde{u}_m^s$ (with possibly additional $O(\eps)$ error contributions).
As any further transition from $V_{s+1}$ to $V_{s}$ costs an additional finite amount
of energy, we have thus obtained a construction which does not deteriorate the
energy of the original construction (i.e. satisfies the estimate \rf{eq:enm}) and its gradient attains an appropriate 
deformation from the wells $SO(2)U_0 \cup SO(2)U_1$ on the right and left boundaries of the domains.

Finally, we can translate the function $\tilde{u}_{m}^s$  such that the internal layer
$[l_n^{\alpha,s},r_n^{\alpha, s}]$ is shifted to $[-i_n^s,i_n^s]$, where
$i_n^s:=(r_n^{\alpha,s}-l_n^{\alpha,s})/2$. This does not change the energy.
Therefore, all statements of the lemma follow. 
\end{proof}

\begin{remark}
We remark that the modified sequence of Lemma \ref{lem:enI} which was obtained via the ``cutting procedure'' summarized in Remark 4.1
does not necessarily preserve the shift $r_s$ in the definition of the boundary layer $C(V_s,V_{s+1},r_s)$. However, for our further purposes it
is only necessary that the resulting energy does not exceed an estimate of the form \rf{eq:enn}.
\end{remark}

\emph{Step 3: Conclusion.}
Applying the previous step with an appropriate choice of $\alpha,\beta$, 
we obtain a sequence $\{\tilde{u}_{m}^s\}_{m\in \N}$ which satisfies the
desired energy estimate
\begin{align}
H_{m}^1\left(\lambda_m \tilde{u}_m^s\left( \frac{\cdot}{\lambda_m} \right) \right) \leq \left\{ \begin{array}{ll} 
B^+(F_{\lambda},V_1,r_{0})+ c\epsilon &\mbox{ if } s=0,\\
C(V_s,V_{s+1},r_s) +c\epsilon &\mbox{ if } s \in \{1,...,K-2\}, \\
B^-(V_{K-1},F_{\lambda},r_{K-1}) +c\epsilon &\mbox{ if } s=K-1.
\end{array} \right.
\lb{eq:enm}
\end{align}
As the gradients of $\tilde{u}_{m}^s$ only deviate from the deformations
 $V_s,V_{s+1}$ or the boundary data $F_{\lambda}$ on horizontal transition layers of the size $O(n^{\alpha})$, we may further translate and patch together
the functions $\tilde{u}_{m}^s$ such that the resulting function $u_{m}$
\begin{itemize}
\item is continuous and attains the desired boundary data,
\item still contains translated versions of the boundary and internal layers in the stripes $[-i_n^s,i_n^s]\times  [-m,m]$,
\item the individual jumps in its gradient converge to the respective jumps of the gradient of $u$,
\item due to \rf{eq:enn} satisfies the overall energy bound 
\bes
H^1_{m}(u_{m})\le H_{surf}(u)+c\eps
\ees
for all sufficiently large $m$.
\end{itemize}
More precisely, we define
\begin{align*}
u_{m}^{i-jj}:= \left\{ \begin{array}{ll} 
F_{\lambda}\begin{pmatrix} (i-j)\lambda_m \\ j\lambda_m \end{pmatrix} & \mbox{ for } i\leq -m,\\
\lambda_m\tilde{u}^0_m\left(\begin{pmatrix} i-j\\ j\end{pmatrix}+\begin{pmatrix}m\\ 0\end{pmatrix}\right)
+ D^{0}_m & \mbox{ for } -m\leq i \leq -m+h^{0}_m,\\
\lambda_m\tilde{u}^s_m\left(\begin{pmatrix} i-j\\ j \end{pmatrix}-\begin{pmatrix}k_m^s\\ 0\end{pmatrix}\right)
+D^{s}_m& \mbox{ for } k^{s}_m -h^{s-1}_{m}\leq i \leq k^{s}_m + h^{s}_m\\
&\mbox{ and } s\in \{1,...,K-2\},\\
\lambda_m\tilde{u}^{K-1}_m\left(\begin{pmatrix} i-j \\ j  \end{pmatrix}-\begin{pmatrix}m\\ 0\end{pmatrix}\right)
+ D^{K-1}_m& \mbox{ for } m-h^{K-1}_m\leq i \leq m,\\
F_{\lambda}\begin{pmatrix} (i-j)\lambda_m \\ j\lambda_m \end{pmatrix} +D^{K}_m& \mbox{ for } i\geq m.
\end{array} \right.
\end{align*}
Here, the off-sets $D^{s}_m$ are defined by  
\begin{align*}
D^{0}_m &:= F_{\lambda}\begin{pmatrix} -1 \\ 0 \end{pmatrix} - \lambda_m \tilde{u}^{0} \begin{pmatrix} 0 \\ 0 \end{pmatrix} ,\\
D^{s}_m &:= 
\lambda_m \tilde{u}^{s-1}_m \begin{pmatrix} h^{s-1}_m \\ 0 \end{pmatrix}-\lambda_m\tilde{u}^{s}_m\begin{pmatrix} -h^{s-1}_m \\ 0 \end{pmatrix}+D^{s-1}_m
\quad\text{for}\quad s\in \{1,...,K-1\},\\
D^K_m &:= \lambda_m\tilde{u}^{K-1}_m\begin{pmatrix} 0 \\ 0 \end{pmatrix}-F_{\lambda}\begin{pmatrix} 1 \\ 0 \end{pmatrix}+ D^{K-1}_m.
\end{align*}
Due to the fixed boundary data on the left lateral side of the parallelogram $\Omega$, i.e. $u_m(x)=F_{\lambda}x$ for $x\leq -1$, the $L^4$ gradient convergence and the fundamental theorem of calculus, i.e. for almost every $y\in [-1,1]$
\begin{align*}
u_m(1-y,y) - u_m(-1-y,y) = \int\limits_{0}^{1} \nabla u_m(t(1-y)-(1-t)(1+y),y)\cdot \begin{pmatrix} 2 \\ 0 \end{pmatrix} dt, \\
\int\limits_{0}^{1} \nabla u(t(1-y)-(1-t)y,y)\cdot \begin{pmatrix} 2 \\ 0 \end{pmatrix} dt=F_{\lambda}\begin{pmatrix} 2 \\ 0 \end{pmatrix} \mbox{ and } \nabla u_m \rightarrow \nabla u \mbox{ in } L^{4}(\Omega),
\end{align*}
we infer that $u_m \rightarrow u$ in
$L^{\infty}$. In particular, the value of $D_m^K$ can be estimated by $o(1)$ as $m\go\infty$. Thus, changing the volume fractions of $V_1, V_2$
slightly (i.e. on a set of measure $o(1)$), it is possible to arrange
$D^K_m=0$. Since this can be obtained via a slight extension/ shortening of
one of the domains on which $\nabla u_m = V_1$ or $\nabla u_m = V_2$, it
does not change the energy. As this preserves the convergence $u_m \rightarrow u$ in $L^{\infty}$, the above modification of the domains $\Omega_s$ allows us to chose the associated, slightly modified function $u_m\in W^{1,\infty}_0(\Omega)+F_{\lambda}x$ as the desired recovery sequence.
Since this construction can be achieved for any $\epsilon>0$ and as any $m\in \N$
can be obtained via the procedure introduced in Lemma 4.4, the claimed $\Gamma$-$\limsup$ inequality follows from a diagonal argument.
\end{proof} 

\begin{figure}[t]
\includegraphics[width=0.9 \linewidth]{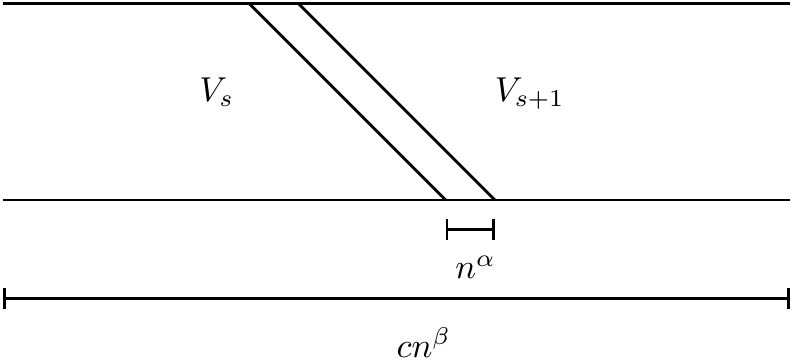}
\caption{During the construction of the recovery sequence in Lemma 4.4, we only keep an
inner core of the scale $O(n^{\alpha})$ of the original function $u_{\bar{n}}$ and paste the deformations $V_s$ and $V_{s+1}$ to the left and
right of it on the length-scale $m(n)=O(n^{\beta})$.}
\end{figure}

As a corollary we obtain
\begin{corollary}
Let $\{u_{n}\}_{n\in \N}$ be a minimizing sequence of $H_{n}(\cdot)$ corresponding to admissible boundary data (\ref{eq:newBC}). Then the total number of boundary and internal layers is equal to three if $\lambda\notin\{0,1\}$. The internal layer is either positioned at the horizontal coordinate $x=1-2\lambda$ or at $x= -1+2\lambda$.
\end{corollary}
\begin{proof}
This follows from the previous $\Gamma$-convergence result, in particular it is implicitly present in Lemma 4.3. For a fixed $\lambda$ the position of the
internal layer is uniquely prescribed by the points $x=1-2\lambda$ or $x= -1+2\lambda$.
\end{proof}

\section{Comparison with Numerics}

In this section we present the results of our numerical simulations finding
local minimizers to (\ref{DC}) considered with fixed $\tau_i\equiv\tau$ 
(i.e. the minimization was done only among chains that are uniform in the vertical direction) and boundary conditions
\be
u^i=U_0\left(\begin{array}{c}i\\0\end{array}\right)\ \text{if}\ i\le -n\ \text{and}\ \ \text{and}\ 
u^i=QU_1\left(\begin{array}{c}i\\0\end{array}\right)\ \text{if}\ i\ge n.
\lb{eq:newBCc}
\ee
The numerics are based on a local optimization Newton type algorithm. Starting with a deformation $u_n\in \mA_{n,\tau}$
corresponding to a martensitic twin, i.e. a configuration such that 
\begin{equation}
u_n^i=U_0\left(\begin{array}{c}i  \\0\end{array}\right)\ \text{if}\ i\le 0\ \ \text{and}\ \
u_n^i=QU_1\left(\begin{array}{c} i \\0\end{array}\right)\ \text{if}\ i\ge 0,
\label{MT}
\end{equation}
we initially preoptimize the position of the middle atom $i=0$ in (\ref{MT}). Namely,
we optimize its position w.r.t. (\ref{DC}) considered with $\tau_i\equiv\tau$ without changing the configuration of the other atoms. 
\begin{figure}
\includegraphics[width=0.6\textwidth]{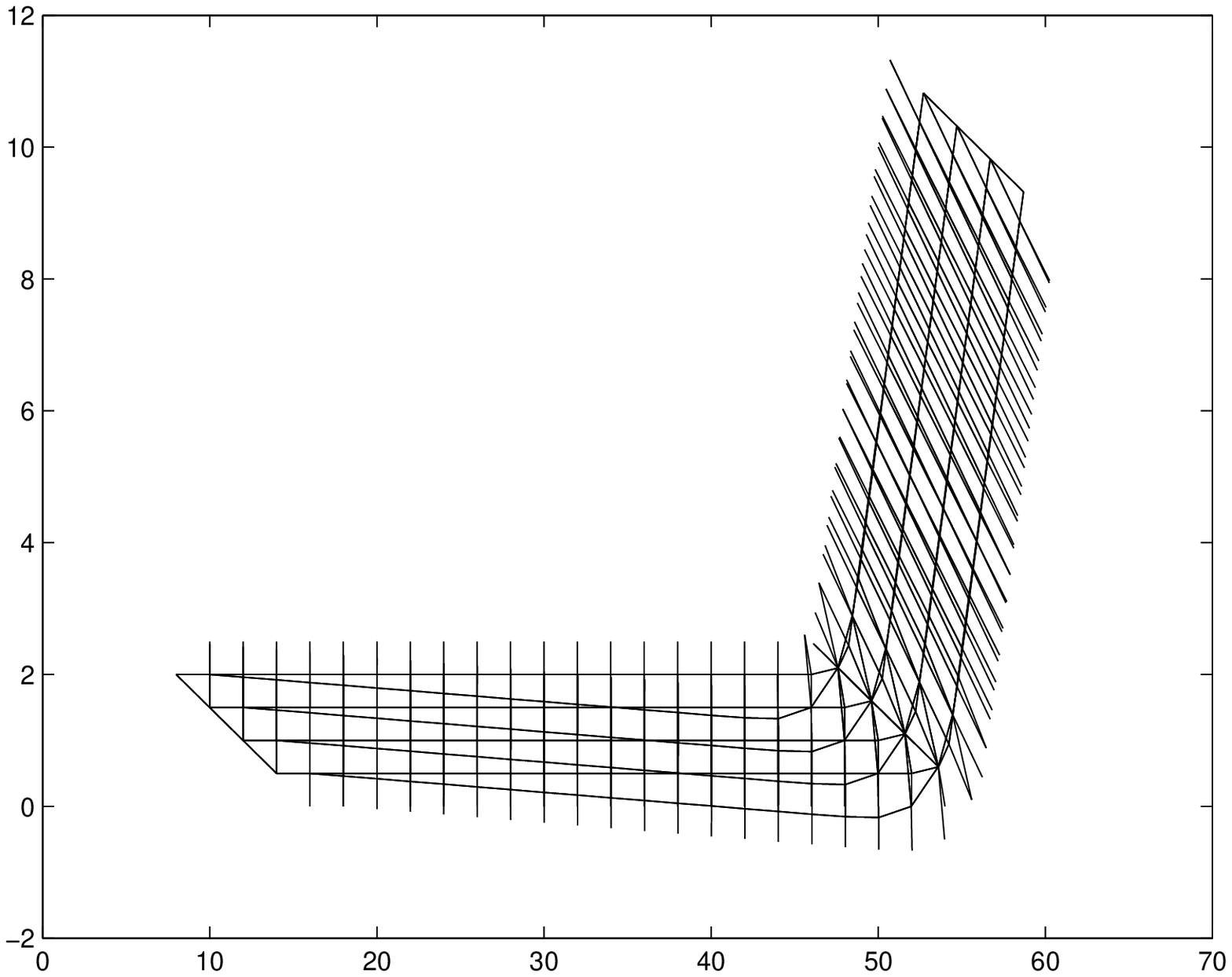}
\includegraphics[width=0.6\textwidth]{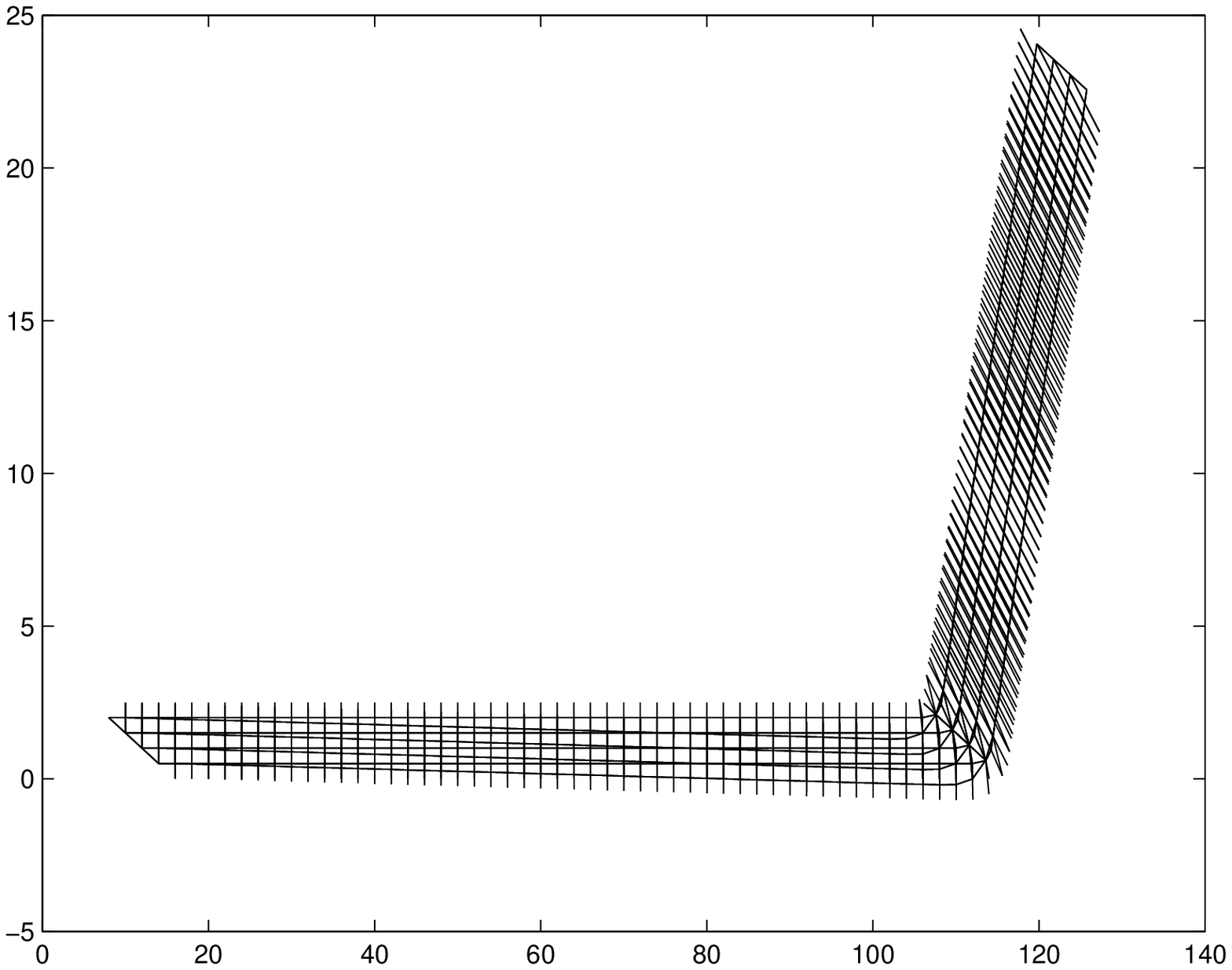}
\caption{Twins appearing as local minimizers ($n=40$ and $n=100$).}
\end{figure}

Next, the resulting preoptimized configuration is used as an initial guess for the Newton
algorithm and thus, a nearby lying local minimizer of (\ref{DC}) considered with (\ref{eq:newBCc}) is found.\\

The results in Fig. 4 show  such a minimizer possessing a straight twinning interface
coinciding with that of (\ref{MT}) and prescribed by the vector $\tau$. 
The figure shows that the deviation from the the twinned configuration quickly decreases as the number of atoms, $n$, tends to infinity. 
\begin{figure}
\includegraphics[width=0.6\textwidth]{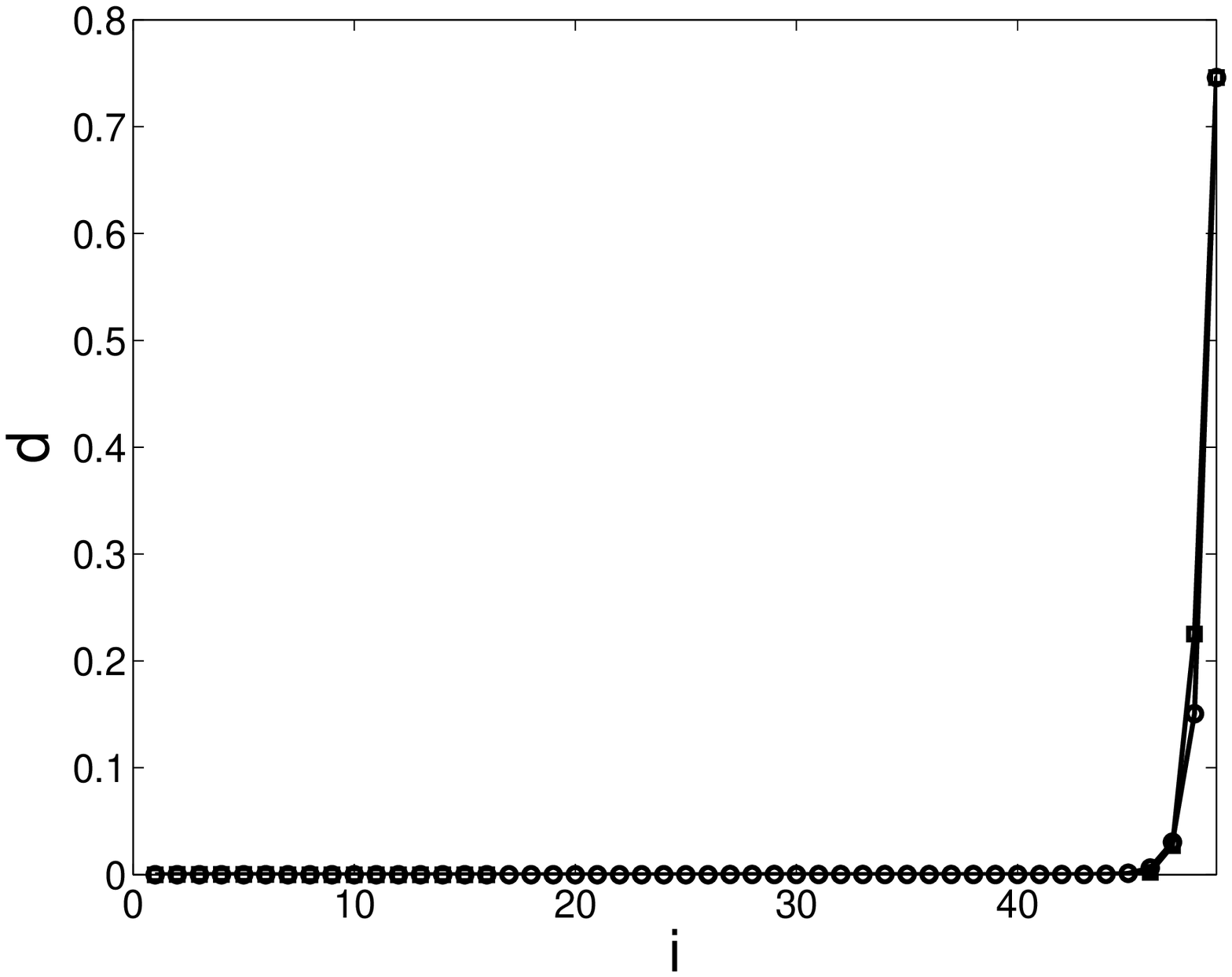}
\includegraphics[width=0.6\textwidth]{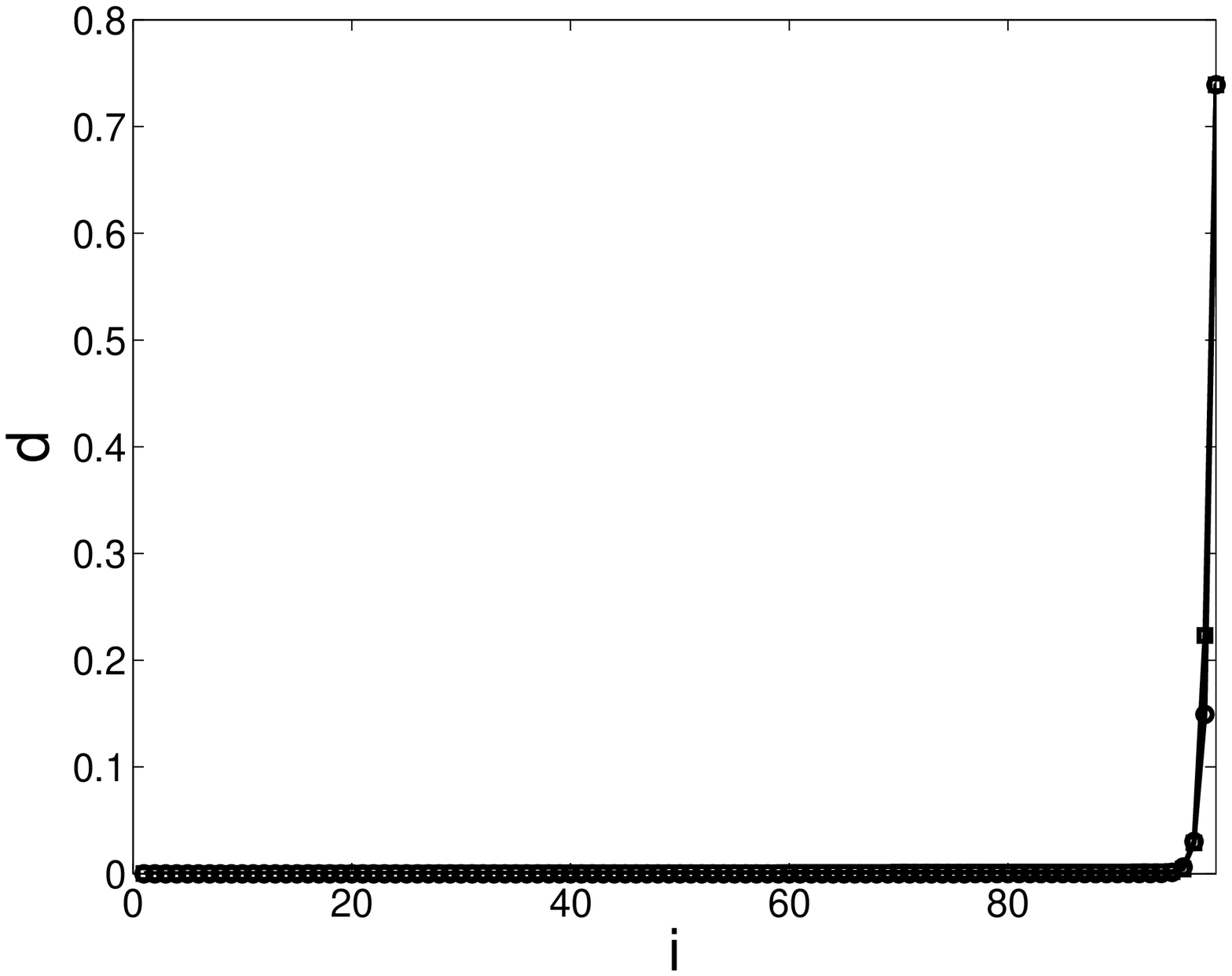}
\caption{Absolute deviation from the rank-one connected twin configuration (\ref{MT}) ($n=100$ and $n=200$).
Square points-- the computed numerical deviation, round points-- its approximation by an exponential profile.}
\end{figure}
In Fig. 5 we plot the absolute deviation in the atom positions of our local minimizer from the rank-one twin configuration (\ref{MT}).
One finds that the deviation decreases exponentially
starting from the middle atom lying on the twinning interface. Note, that the decay is not given exactly by a simple exponent but nevertheless is nicely approximated by it. Surprisingly, the middle atoms of our local minimizer and the initial guess for the Newton algorithm 
lying on the twinning interface coincide (deviation between them is zero). 
This might appear due to the initial preoptimization of the position of the middle atom
in (\ref{MT}) which was described above. This simple preoptimization step seems to find the right position of the middle atom on the twinning interface.

\section{Summary and Discussion}
In this article we introduced a general type of two-well Hamiltonian defined on a two-dimensional sublattice of $\Z^2$ by imposing the assumptions (H1)-(H4).
After restricting the set of possible deformations to the special case of 1D chains, non-uniformly 
extended in the vertical direction and considered with the boundary
conditions \rf{eq:newBC}, we were able to show piecewise asymptotic rigidity of sequences whose energy scales as surface energy. The corresponding compactness and $\Gamma$-convergence arguments allowed us to rigorously derive the continuum limit of the surface energy concentrating on the line interfaces between twin configurations. Finally, a numerical minimization of the discrete problem reflected our analytical results and showed
an interesting exponential decay in the boundary layer profiles between arising twins.\\

Keeping these results in mind, we conclude by briefly commenting on the underlying physical assumptions, possible generalizations and some interesting related questions:\\

Since low energy states are expected to remain close to laminar
configurations, our class of constrained configurations, i.e. atomic chains, seem to be natural objects -- even though they impose restrictions on the model. An immediate -- though less natural -- generalization to the three-dimensional two-well problem is possible: Considering configurations in which a chain of atoms (i.e. the atoms on the (i,0,0)-line with $i\in [-n,n]$) is freely deformed while the atoms on the corresponding orthogonal two-dimensional planes are deformed with a variable elongation, $\tau^{i}_{n}$, in one direction and a fixed extension, $\tau$, in the other planar direction, (basically) reduces this 3D situation to our 2D situation. Indeed, under these assumptions (and appropriate Dirichlet boundary conditions) the 3D setting corresponds to rank-one perturbations of a ``one-dimensional'' configuration. As in our two-dimensional framework this then allows to conclude that in the in-plane directions all deformations have to be close to a single well -- jump!
 s between the wells are impossible in this direction. This again relies on the fact that there are at most two intersections with the wells along any arbitrary rank-one direction in the matrix space.\\
As our arguments rely on the one-dimensionality in the direction vertical to
the generating chain, it is at the moment neither clear how to extend our
results to the full two-dimensional setting nor to the three-dimensional case
with variable elongations in both of the planar directions (i.e. which in a
sense would correspond to a ``$(1+2\epsilon)$-dimensional'' argument). \\

In the case of general boundary conditions one expects that minimizers reflect the microstructure predicted in continuum theories and determine a length scale for the microstructures. For an investigation of this question one would need to proceed to the full two-dimensional setting which seems to be a very difficult open problem which is not even fully understood in the continuous framework. \\

Finally, an analytical identification of the minimizing sequences in the definitions of the boundary and internal layers
\rf{BL}--\rf{IL} poses a further interesting problem. It seems impossible to find explicit solutions
of the underlying Euler-Lagrange systems -- even for our model Hamiltonian \rf{HD}.
Nevertheless, it could be possible to justify the exponential decay of the boundary and internal layers found numerically in Fig. 4--5 following e.g.
approaches outlined in~\cite{CT02,Hu13}. From an analytical side already the ``cutting procedure'' introduced in Lemma 4.2 and Remark 4.1
shows that the width of the internal and boundary layers in the corresponding infimizing sequences 
can be made arbitrarily algebraically small, i.e. of the size $O(n^\alpha)$ for any $0<\alpha<1/2$. This again suggests that
the width of the layers should decay exponentially with $n$.

\section*{Acknowledgments} 
G.K. acknowledges the postdoctoral scholarship at the Max-Planck-Institute for Mathematics in the Natural Sciences, Leipzig. A.R. thanks the Deutsche Telekom Stiftung and the Hausdorff Center of Mathematics for financial support. Furthermore she would like to thank the MPI for its kind hospitality.
The authors thank Jens Wohlgemuth for valuable comments.

\bibliography{bibliography}
\bibliographystyle{unsrt} 
\clearpage
\addtocounter{tocdepth}{2}

\end{document}